\documentclass[conference]{IEEEtran}

\usepackage{balance}
\usepackage{booktabs} % For formal tables
\usepackage{algorithmic}
\usepackage{graphicx}
\usepackage{textcomp}
\usepackage[most]{tcolorbox}
\usepackage{xcolor}
\usepackage[ruled]{algorithm2e} % For algorithms
\usepackage{comment}
\usepackage{amsfonts}
\usepackage{amsmath}
\usepackage{amsthm}
\usepackage{tabularx}
\usepackage{multirow}
\usepackage{paralist}
\usepackage{listings}
\usepackage{caption}
\usepackage{colortbl}
\usepackage{tikz}
\usetikzlibrary{arrows.meta,positioning}

\usepackage{subcaption}

\newcolumntype{Y}{>{\centering\arraybackslash}X}

\newcommand*\darkcircled[1]{%
  \tikz[baseline=(char.base)]%
    \node[shape=circle, draw=black, fill=black, text=white, inner sep=.8pt] (char) {#1};}

\usetikzlibrary{calc}
\newcommand{\tikzmark}[1]{%
  \tikz[remember picture,overlay]\node[inner sep=0pt,outer sep=0pt](#1){};}
    
\newtheorem{defn}{Definition} % definition numbers are dependent on theorem numbers

\newtheorem{theorem}{Theorem}

\newcommand{\frmtitle}{Safety Graph (SG)}

\newcommand{\frm}{SG}

\newcommand{\pratt}[2]    {\ensuremath{\mathit{pre}^{#2}({#1})}}
\newcommand{\postp}[2]    {\ensuremath{\mathit{post}^{+#2}({#1})}}
\newcommand{\postn}[2]    {\ensuremath{\mathit{post}^{-#2}({#1})}}

\newcommand{\src}    {\ensuremath{\mathit{src}}}
\newcommand{\dst}    {\ensuremath{\mathit{dest}}}
\definecolor{pddlkeyword}{RGB}{128,0,96}
\definecolor{pddlcomment}{RGB}{110,110,110}

\lstdefinelanguage{PDDL}{
    alsoletter={:,-},
    morekeywords={:action,:parameters,:precondition,:effect},
    sensitive=true,
    morecomment=[l]{;}
}

\lstdefinelanguage{PDDLPlan}{
    alsoletter={-},
    sensitive=true,
    morecomment=[l]{;}
}

\lstdefinestyle{pddlplanstyle}{
    language=PDDLPlan,
    basicstyle=\ttfamily\tiny,
    commentstyle=\color{orange}\itshape,
    showstringspaces=false,
    breaklines=true,
    columns=fullflexible,
    keepspaces=true,
    tabsize=2,
    frame=none,
    escapeinside={(*@}{@*)}
}

\lstdefinestyle{pddlstyle}{
    language=PDDL,
    basicstyle=\ttfamily\tiny,
    keywordstyle=\color{pddlkeyword}\bfseries,
    commentstyle=\color{pddlcomment}\itshape,
    showstringspaces=false,
    breaklines=true,
    columns=fullflexible,
    keepspaces=true,
    tabsize=2,
    frame=none
}

\lstdefinelanguage{PDDL}{
  morekeywords={define, domain, problem, requirements, strips, typing,
                negative-preconditions, conditional-effects,
                types, predicates, action, parameters, precondition,
                effect, when},
  sensitive=false,
  morecomment=[l]{;},
  morestring=[b]"
}
 
\lstdefinestyle{valoutputstyle}{
    basicstyle=\ttfamily\scriptsize,
    numbers=none,
    breaklines=false,
    escapeinside={(*@}{@*)},
    columns=fullflexible,
    keepspaces=true,
    aboveskip=0pt,
    belowskip=0pt,
    xleftmargin=0pt,
    frame=none
}

\ifCLASSINFOpdf
\else
\fi
\begin{document}
%
% paper title
% Titles are generally capitalized except for words such as a, an, and, as,
% at, but, by, for, in, nor, of, on, or, the, to and up, which are usually
% not capitalized unless they are the first or last word of the title.
% Linebreaks \\ can be used within to get better formatting as desired.
% Do not put math or special symbols in the title.
\title{Safety Graph: When Safety Responses Become Attack Surfaces in Cyber-Physical Systems}

% author names and affiliations
% use a multiple column layout for up to three different
% affiliations
% \author{\IEEEauthorblockN{Rakesh Podder$^{\dagger}$$^*$, Philip D. Hopkins$^{\dagger}$, Shadaab Kawnain Bashir$^{\dagger}$, Sarath Sreedharan$^{\dagger}$, Indrajit Ray$^{\dagger}$$^*$, \\ and Indrakshi Ray$^{\dagger}$}
% 	\IEEEauthorblockA{$^{\dagger}$Colorado State Uinersity, Fort Collins, Colorado, 80521, USA\\
% 		\{Rakesh.Podder, Philip.Hopkins, Shadaab.Bashir, Sarath.Sreedharan, Indrajit.Ray, Indrakshi.Ray\}@colostate.edu \\ $^*$ Corresponding Authors}
%             }
	% \and
	% \IEEEauthorblockN{Homer Simpson}
	% \IEEEauthorblockA{Twentieth Century Fox\\
	% 	homer@thesimpsons.com}
	% \and
	% \IEEEauthorblockN{James Kirk\\ and Montgomery Scott}
	% \IEEEauthorblockA{Starfleet Academy\\
	% 	someemail@somedomain.com}

\author{
\IEEEauthorblockN{Rakesh Podder, Philip D. Hopkins, Sarath Sreedharan, Indrakshi Ray, and Indrajit Ray}
\IEEEauthorblockA{Computer Science Department, Colorado State University, Fort Collins, Colorado, USA \\ 
\{Rakesh.Podder, Philip.Hopkins, Sarath.Sreedharan,  Indrakshi.Ray, Indrajit.Ray\}@colostate.edu}
\\ 
}

\maketitle

% As a general rule, do not put math, special symbols or citations
% in the abstract
\begin{abstract}

Modern industrial control systems (ICS) are susceptible to cyber attacks triggering safety failures in the operational technology (OT) layer. The traditional way to mitigate any failure in  OT is to implement safety responses.
However, we show in this work that these responses can open up new attack surfaces, thereby vastly expanding the scope of lateral movement.
Existing frameworks model cyber-to-physical escalation but provide no mechanism to reason about how these safety responses to OT-failures can trigger new cascading cyber attacks in the Information Technology (IT) layer. To enable such reasoning, we introduce \textit{Safety Graph} (SG), a formal framework unifying three paradigms -- Attack Graphs (AG), Fault Trees (FT), and Event Trees (ET) -- over a shared propositional state space, capturing the full bidirectional IT-OT cascade. We provide a sound and complete compilation of SG to a classical AI planning problem, performing plan search using A* to find cross-layer cascade paths.
Through a Web3D modeling and simulation of the Tennessee Eastman chemical process plant, which is a widely used chemical engineering benchmark serving as a standard for testing control, monitoring and fault propagation,
we demonstrate that the planning compilation of SG automatically discovers \textit{safety-response-induced} lateral movements, and identifies the responsible barriers that are unidentifiable under existing models.
A complexity and scalability analysis characterizes how the planning compilation of SG scales with model size and cascade depth.
To the best of our knowledge, SG is the first framework to formally model IT-OT attacks and enable operators to explore the problem of cyber attacks triggering cascading failures in ICS.
\end{abstract}

\begin{IEEEkeywords}
Cyber-Physical System, Attack Graph, Fault Tree, Event Tree, AI Planning
\end{IEEEkeywords}
% no keywords

% For peer review papers, you can put extra information on the cover
% page as needed:
% \ifCLASSOPTIONpeerreview
% \begin{center} \bfseries EDICS Category: 3-BBND \end{center}
% \fi
%
% For peerreview papers, this IEEEtran command inserts a page break and
% creates the second title. It will be ignored for other modes.
\IEEEpeerreviewmaketitle

% -----------------------
% ------------------------
\section{Introduction}
Modern industrial infrastructure depends on the tight coupling of Information Technology (IT) and Operational Technology (OT). IT systems handle enterprise communication, data management, and business logic. OT systems — Programmable Logic Controllers (PLCs), Human-Machine Interfaces (HMIs), and SCADA servers — control physical processes in power grids, water treatment plants, chemical refineries and pipelines. %Together, these form Cyber-Physical Systems (CPS), where digital decisions produce physical consequences, and physical faults reshape the attack surface.

This coupling creates security challenges %that goes beyond conventional network defense 
as IT and OT teams operate in silos (NIST IR 8576 Initial Public Draft \cite{tang2026transit}). A successful intrusion in a CPS does not merely expose data — it can manipulate pressure setpoints, disable safety valves, or corrupt control logic, with consequences that are physical, potentially irreversible, and catastrophic. Major CPS attacks such as the 2010 Stuxnet attack \cite{baezner2017stuxnet}, the 2016 Ukraine power grid \cite{case2016analysis} or the 2021 Colonial Pipeline attack \cite{beerman2023review}, all share a consistent pattern where attackers enter through IT and cascade downward into OT.

%propagated through IT infrastructure and reprogrammed Siemens PLCs to destroy uranium centrifuges while reporting normal operation to operators. The 2016 Ukraine power grid attack used spear-phishing to compromise IT workstations before pivoting into OT to cut power to over 200,000 households \cite{case2016analysis}. The 2021 Colonial Pipeline ransomware attack disabled fuel supply across the US East Coast, originating entirely in the IT . These incidents 

The research community has responded to solve these challenges. A substantial body of work models how cyber attacks propagate from IT to OT — using attack graphs \cite{barrere2023cyber,hawrylak2012using,ibrahim2019attack}, attack trees \cite{lemaire2017security,xie2013security}, resiliency graph \cite{bashir2024resiliency}, state dependency graphs \cite{adamos2024enhancing}, and hybrid frameworks that link cyber and physical layers \cite{semertzis2022quantitative,sahu2023inferring}. These models capture the cyber incidents and their effects on a physical failure. %This direction is well-studied.
However, the mitigation solutions of physical failures do not simply terminate the threat. They trigger safety responses — alarms, operator procedures, diagnostic checks, maintenance modes — that alter the system's operational state. As we observe in this work, each such response is a potential new cyber exposure. A software vulnerability that is unexploitable under normal operating conditions may become exploitable after a network reconfiguration is performed during fault recovery. %A workstation that is architecturally unreachable may be brought online by a technician diagnosing a PLC failure. 
%In other words, the physical fault, and the safety procedure it initiates, creates new attack surfaces in the IT layer that are not considered before.

This OT-to-IT jump, where a physical fault enables a new cyber attack, is precisely what existing models fail to capture and enable analysis. 
We identify this as a new class of lateral movement that we call \textit{safety-response-induced lateral} movement and is a subtle exploitable weakness in ICS environments. %Conventional lateral movement exploits network reachability and vulnerabilities. 
Safety-response-induced lateral movement exploits transient operational state changes — \textit{diagnostic connectivity},\textit{ maintenance mode activation}, \textit{temporary host exposure} — that occur as a direct consequence of prescribed safety procedures. %An attacker who can predict or induce these state changes gains access to hosts that are unreachable under normal operating conditions, without exploiting any additional vulnerability.
Current frameworks model cyber-to-physical triggering in detail, but treat safety responses as terminal events rather than as intermediate states that reshape the cyber threat landscape. Even the most expressive prior frameworks handle this asymmetrically, as cyber behavior is modeled explicitly, while physical and safety-response consequences are abstracted into discrete end states.
IT professionals have used Attack Graphs (AG) to analyze system vulnerabilities and plan cybersecurity defenses \cite{lallie2020review}. ICS operators have separately used Fault Trees (FT) to reason about the propagation of failure \cite{ruijters2015fault}, and Event Trees (ET) to model time-ordered safety responses that follow a fault \cite{ericson2015hazard}. No prior framework captures the full bidirectional cascade — IT attacks triggering OT faults (IT $\to$ OT) and safety responses creating new IT attack surfaces (OT $\to$ IT).

To address this gap, we introduce \frmtitle, a formal modeling framework that integrates AG, FT, and ET into a unified propositional state space. Cross-domain dependencies are encoded through a shared attribute set over all hosts where an attack's postconditions can enable a fault transition, a fault's occurrence can trigger an event, and event barriers (safety responses) outcome can establish new host-level cyber preconditions that enable subsequent attacks. This captures the full bidirectional cascade without requiring separate linking relations between layers.
However, modeling alone is not sufficient for practical security analysis. A defender operating on a real IT-OT deployment cannot manually enumerate all feasible cross-layer paths — the space is too large, and the interactions between cyber, physical, and safety-response layers are too subtle. We address this through a sound compilation of \frm\ into classical AI planning. Every attack step, fault transition, and event-tree barrier evaluation become a planning action. Every system state proposition becomes a planning fluent. The question of whether a critical failure is reachable under a given IT-OT configuration is reduced to a planning problem, which an AI planner can solve automatically. %AI planning is known to scale to very large problem spaces \cite{podder2025spear,bashir2024resiliency} — this is a key advantage of our approach over manual graph inspection or exhaustive model checking.

To evaluate \frm\ in practice, we build IT-OT test-benches modeled on the Tennessee Eastman chemical process plant \cite{chiang2001tennessee}, with a real network architecture implemented on virtual machines (VMs). We instantiate an AG from network configuration, cyber threat intelligence (CTI) reports, and penetration-testing evidence, construct an FT with input from domain experts, and derive ETs from safety manuals. The AT, FT and ET are the inputs to our modeling process. All three are automatically compiled into a planning model (encoded in Planning Domain Definition Language (PDDL)) using the compilation defined in $\S$\ref{sec:SG}. %The AG captures the reachable cyber state space over hosts; the FT captures fault escalation conditions and their enabling propositions; the ET captures time-ordered safety barriers whose outcomes introduce new host-level propositions into the shared state space. 
The AI planner operates over this unified model to discover feasible cross-layer cascades. %When analyzed using only the AG and FT layers, the planner finds a valid plan triggering an initial fault. However, critical faults requiring cyber preconditions absent from the reachable state space remain undiscovered. 
%Once the ET is integrated to form \frm\, it gains barrier-evaluation actions whose postconditions introduce new attack actions, extending the reachable state space across all three layers. 
The planner discovers a path traversing AG $\to$ FT $\to$ ET $\to$ AG $\to$ FT, establishing that the critical failure is reachable under \frm. %, but provably absent from the AG+FT-only model. 
We validate every step of the planner-generated trace in the live test-bench. 
In this work, we make the following major contributions:

\begin{enumerate}[i)]
\item We develop a uniform framework (namely, safety graph, SG) that allows one to model security, faults, and safety, and their interdependencies in a cyber-physical system. To the best of our knowledge, ours is the first effort in this direction.
\item We compile SG to a  classical planning problem and provide theoretical results proving soundness, completeness, and other theoretical guarantees.
\item We demonstrate how to leverage AI planning for reasoning about attacks, faults, and safety events in SGs.
\item We also identify and formalize safety-response-induced lateral movement, a new attack class in which the prescribed safety response to a physical fault opens up a new attack surface. Any model that treats safety responses as terminal is unable to capture this class of attacks.
\item We demonstrate our approach on Tennessee Eastman; specifically, we show that the critical failure is unreachable using AG and FT, but is reachable in SG; we identify the enabling barrier, and validate every planner step in a live test bench with real tooling.

\end{enumerate}

The rest of this paper is organized as follows. Section \ref{sec:related-work} surveys related work. Section \ref{sec:runnning} presents the running example and test-bench construction. Section \ref{sec:threat-model} provides the threat model and assumptions. We formalize \frm\ in Section \ref{sec:SG}. Section \ref{sec:analysis} describes the workflow and analysis on running example. Section \ref{sec:eval} reports the evaluation. Section \ref{sec:con} concludes the findings and future works.
%We have evaluated our proposed modeling with a second example, an IEEE 14-bus power grid in Appendix \ref{app:power-grid}. 
% ------------------------
% \section{Background}
% \input{20-Background}
% ------------------------
\section{Related Work}

CPS security research has largely focused on  one direction -- attacks that originate in the IT layer and cascade  downward into operational technology -- ultimately disrupting or manipulating the physical process. The dominant tools for this analysis are the attack graph  \cite{barrere2023cyber,hawrylak2012using,ibrahim2019attack,talukder2025vkg2ag}, and a rich body of work has developed around it \cite{lemaire2017security,xie2013security,semertzis2022quantitative,sahu2023inferring,podder2026exam,podder2026sera}. %\textcolor{blue}{Provide citations - Done}

% Barrère et al. introduced Cyber-Physical Attack Graphs (CPAG) \cite{barrere2023cyber}, extending traditional attack graphs to represent the interconnected relationships between cyber and physical components in ICS. Similarly, Hawrylak et al. proposed Hybrid Attack Graphs (HAG) for the smart grid \cite{hawrylak2012using}, explicitly encoding physical parameters alongside cyber states so that cyber actions that indirectly cause physical failures can be represented. Ibrahim et al. model attacks and their effects through attack graphs that capture attacker privileges, system vulnerabilities, and state transitions, and employ model checking to verify security properties across jointly modeled cyber and physical elements \cite{ibrahim2019attack}. Adamos et al. assess CPS attack resilience using state dependency graphs \cite{adamos2024enhancing} that trace how cyber exploits propagate through component states, creating cascading attack paths across both IT and OT layers.

% Beyond attack graphs, Lemaire et al. automatically generate system- and attacker-specific attack trees \cite{lemaire2017security} from a formal CPS architecture model, representing cyber and physical components together at the architectural level. Xie et al. similarly use attack trees \cite{xie2013security} to quantitatively evaluate threat, vulnerability, and expected damage along attack paths, integrating cyber and physical factors within the same tree structure.

Barrère et al. extend attack graphs to Cyber-Physical Attack Graphs (CPAG), linking cyber and physical components in ICS \cite{barrere2023cyber}. Hawrylak et al. encode physical parameters alongside cyber states in Hybrid Attack Graphs for the smart grid, capturing cyber actions that cause physical failures \cite{hawrylak2012using}. Ibrahim et al. apply model checking over graphs of attacker privileges, vulnerabilities, and state transitions \cite{ibrahim2019attack}. Adamos et al. use state dependency graphs to trace cascading paths across IT and OT layers \cite{adamos2024enhancing}. Attack trees address the same problem. Lemaire et al. generate system- and attacker-specific trees from a formal CPS architecture model \cite{lemaire2017security}. Xie et al. quantify threat, CVEs, and expected damage along paths within one tree \cite{xie2013security}.

Several works have addressed quantitative risk in CPS. Semertzis et al. use attack graphs to model attacker 
progression in electric power systems  \cite{semertzis2022quantitative} and perform quantitative risk analysis. %, with physical impacts evaluated through a co-simulated digital twin. 
Presekal et al. propose a modeling approach for online 
situational awareness in power grid OT networks \cite{presekal2023attack}, detecting early-stage cyber kill-chain activity through anomaly detection and mapping behaviors into attack graphs. This work is extended \cite{presekal2025anomaly} to combine hybrid deep learning with spatial-temporal traffic features to detect and isolate attack nodes in the power grid. Few et al. present a case study \cite{few2021case}  using the Purdue model~\cite{williams1994purdue} to show how an attacker can traverse from Level 5 (Internet DMZ) down to Level 0 (Field), capturing multi-layer traversal within a single attack graph.

On the defensive side, Sahu and Davis model adversarial behavior in cyber-physical power systems using a Bayesian attack graph approach \cite{sahu2023inferring} to infer attacker intent and track progression under uncertainty. 
Banik et al. automatically synthesize defense strategies for CPSs \cite{banik2023automated} that account for adaptive adversary behavior. Khaled et al. introduce a framework for modeling industrial CPS under smart, adaptive attackers and recommending countermeasures \cite{khaled2020assessing}, capturing interactions across cyber, control, physical, and social components within a unified formal framework.

% Despite this breadth, the existing literature shares a structural limitation -- the cyber-physical integration is asymmetric. Cyber behavior is modeled explicitly and in detail, while the physical and operational side is treated as a collection of discrete end states or final attack objectives. None of these frameworks model the physical process as a dynamic layer whose faults and safety responses actively reshape the cyber attack surface. 

% The closest work is by Ray et al., who propose the Explainable Resiliency Graph (ERG) \cite{ray2023explainable} to analyze how system vulnerabilities can trigger safety events in a CPS and cause cascading 
% effects within a unified transition system. Building on this, Bashir et al. introduce the Resiliency Graph (RG) \cite{bashir2024resiliency}, which uses AI planning to capture the linkage between vulnerabilities and the system failures they induce, demonstrated on an IT-to-OT cascade. 

% But no existing framework jointly models cyber attacks, physical fault, and safety events in a unified propositional state 
% space. As a result, none can capture 
% scenarios in which IT attacks trigger OT faults (IT$\to$OT), and OT fault-driven safety responses creating new IT attack surfaces (OT$\to$IT). We address this in our current work.

Despite this breadth, the cyber-physical integration of cyber attacks, faults, safety measures remains asymmetric. While cyber behavior is modeled in detail, the physical and operational side appears only as discrete terminal states or final attack objectives. The physical process is never considered as a layer whose faults and safety responses reshape the cyber attack surface.
The closest work is by Ray et al., whose Explainable Resiliency Graph (ERG) analyzes how vulnerabilities trigger safety events and cascade within a unified transition system \cite{ray2023explainable}. Bashir et al. extend this with the Resiliency Graph (RG), using AI planning to link vulnerabilities to the failures they induce, demonstrated on an IT-to-OT cascade \cite{bashir2024resiliency}. Both treat propagation as one-directional. Neither captures the path, in which fault-driven safety responses open new IT attack surfaces (OT $\to$ IT). We address this.

\label{sec:related-work}
% ------------------------
\section{Running Example}
\label{sec:runnning}
% \begin{table*}[!h]
% \centering
% \scriptsize
% \caption{VM Stack for Test-bench}
% \label{tab:vm-stack}
% \renewcommand{\arraystretch}{1.15}
% \begin{tabular}{l|l|l|l|l}
% \hline
% \textbf{VM Name} & \textbf{OS} & \textbf{RAM, CPU, Storage} & \textbf{Components} & \textbf{Description} \\
% \hline
% HMI &  Debian & 1GB, 3, 100GB & ScadaBR & Perform real-time monitoring.\\
% Firewall  & Ubuntu & 512MB, 1, 5GB & pfSense& Separate OT and IT zone \\
% Workstation & Ubuntu & 1GB, 2, 10GB & OpenPLCv2 Editor & To program the PLC \\
% TE Plant  & Ubuntu & 1GB, 2, 5GB & Unity Engine& Web Visualization of TE \\
% PLC   & Ubuntu&512MB, 1, 8GB& OpenPLCv2 Runtime& Perform as a Virtual PLC \\
% Attacker & Kali & 8GB, 4CPU, 50GB & Metasploit and Nmap & To perform attack  \\
% \hline
% \end{tabular}
% \end{table*}

We start our discussion by first describing the ICS test bench of the Tennessee Eastman Process (TEP) plant \cite{chiang2001tennessee, formby2018lowering} using a combination of web-3D modeling and virtualization. TEP is a widely used chemical manufacturing benchmark developed by Eastman Chemical Company  \cite{chiang2001tennessee}. It serves as a standard for testing control, monitoring, and fault propagation, featuring five main units: a reactor, condenser, compressor, separator, and stripper.%, as shown in  as illustrated in Figure \ref{fig:TE-Plant}. 
%We have also evaluated our proposed modeling with another example, an IEEE 14-bus power grid in Appendix \ref{app:power-grid}. 

% \begin{figure}[!h]
%     \centering
%     \includegraphics[width=0.9\linewidth]{figure/TE-Plant-Normal.png}
%     \caption{A Simulated Tennessee Eastman Process (TEP) Plant running under normal condition.}
%     \label{fig:TE-Plant}
% \end{figure}

 \begin{figure}[!h]
    \centering
    \includegraphics[width=0.9\linewidth]{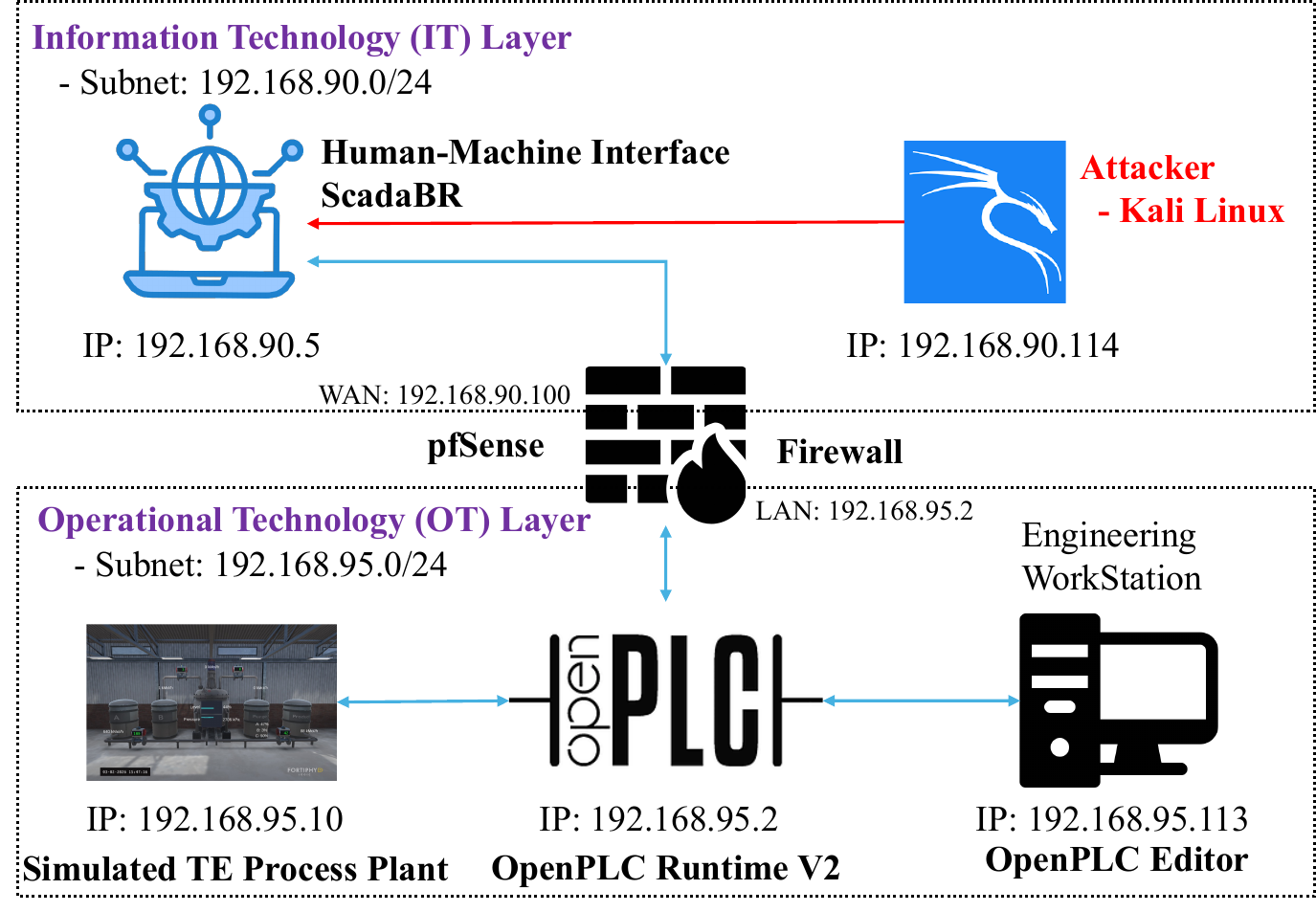}
    \caption{Schematic architecture of the ICS test bench with virtual HMI, firewall, PLC, workstation, and TEP plant simulation.}
    \label{fig:testbench-topology}
\end{figure}

The networked-ICS test-bench illustrated in Figure \ref{fig:testbench-topology} enables us to (i) execute realistic cyber actions, (ii) observe their physical impact on a controlled process simulation, and (iii) evaluate how safety responses can lead to new cyber exposure. 
The test-bench consists of a Human-Machine Interface (HMI) using ScadaBR, an OpenPLC software, an engineering workstation with OpenPLC Editor, a pfSense firewall separating the IT and OT subnets, a process simulation backend, and a Kali Linux attacker VM. %as illustrated in Table \ref{tab:vm-stack}. 
The IT and OT segments use the \texttt{192.168.90.0/24} and \texttt{192.168.95.0/24} subnets respectively. The firewall enforces segmentation between the two networks, while the process simulation and PLC jointly provide an observable control loop on which the effects of cyber actions can be validated. Information related faults and safety measures for TEP is collected from published articles \cite{shang2024tree, agbo2022conflict,reinartz2021extended}.

\section{Threat Model and Assumptions}
\label{sec:threat-model}

\noindent \underline{\textbf{Adversary Model.}} We consider a capability-bounded adversary with an established foothold in the IT layer — acquired potentially via phishing, exposed services, or exploiting known Common Vulnerabilities and Exposures (CVEs) — whose objective is to trigger safety-critical failures in the OT layer. The adversary has the following properties:
\begin{inparaenum}[(i)]
    \item The adversary exploits known vulnerabilities (e.g., published CVEs) in both IT and OT components. Note, zero-day exploits are outside the scope of this model.
    \item The adversary executes sequences of attack steps across networked components, progressing from an initial IT foothold toward OT targets by satisfying the preconditions of each step using postconditions established by prior steps.
    \item Cyber actions executed by the adversary may establish postconditions that satisfy the enabling conditions of OT fault transitions. %The resulting fault cascade and its propagation are governed by physical process dynamics and are not under adversary control. 
    The adversary cannot alter prescribed safety procedures but can corrupt the information environment such as sensor feeds, operator displays etc.\ on which barrier evaluation depends, and can exploit the operational side effects of safety responses as new attack surfaces.
    \item The adversary is constrained by network reachability, host-level access state, and the precondition structure of available attack actions. %Actions whose preconditions are not satisfied in the current system state cannot be executed.
\end{inparaenum}

\noindent \underline{\textbf{Operator Model.}} System operators follow safety procedures prescribed under IEC 61511 \cite{smith2020safety} and ISA-84 \cite{isa84} for each initiating fault event. These procedures are deterministic and documented. Given a fault, the sequence of barriers of event trees, and their outcomes is fixed by certified safety documentation. Operators may modify the IT layer between incidents — patching vulnerabilities, restricting connectivity, or disabling services — and may or may not select the prescribed safety response during a fault event.

\noindent \underline{\textbf{Modeling Assumptions.}}
\begin{inparaenum}[(i)]
    \item Vulnerability information, network configuration, and failure relationships are available from CVE databases. Penetration testing results, system logs, CTI reports, and engineering documentation are provided for each machine in the network. Attack graph (AG) construction, fault tree (FT) specification, and event tree (ET) derivation are grounded from these sources.
    \item All three modeling layers — the AG, FT, and ET — operate over a single global proposition set defined over the hosts and devices of the IT-OT deployment. This shared space is the mechanism by which cross-layer dependencies are captured without explicit linking relations.% i.e., postconditions of cyber attack actions can appear as preconditions of fault transitions; fault postconditions can appear as initiating event propositions that activate event trees; and barrier outcome postconditions can appear as preconditions of subsequent cyber attack actions. 
    This assumption requires that vulnerability descriptions, fault condition labels, and safety procedure states use semantically consistent propositions. Syntax mismatches across layers is addressed by domain expert reconciliation.% and are outside the scope of this work.
    % \item Both cyber attacks and physical fault transitions are expressible as precondition-postcondition tuples. This abstraction enables unified cross-layer reasoning but requires domain expertise for accurate instantiation.
    %\item Given a satisfied precondition set, the effects of attack actions, fault transitions, and barrier evaluations are deterministic. For safety barriers this reflects regulatory practice such as IEC 61511 which mandates exhaustively specified, prescribed responses for each fault scenario. For cyber actions this abstracts away execution uncertainty present in real exploits. This assumption enables classical planning-based analysis and is consistent with the defender's use case — identifying whether a dangerous path exists, not estimating its probability.
    \item The vulnerability and configuration state of the system is fixed at analysis time. Dynamic changes — patch deployment, configuration drift, or new CVE disclosure — are considered as new state which require re-running the analysis. %Incremental re-analysis as the system evolves is left as future work.
    \item  Analysis operates within a bounded relational scope such that there exists a finite set of hosts, a finite fault and event set derived from certified safety documentation, and a bounded number of AG $\to$ FT $\to$ ET cascade cycles. %Reasoning over unbounded cascade cycles is left as future work.
\end{inparaenum}

% ------------------------
\section{\frmtitle\ Formalization}
\label{sec:SG}

We now introduce \frmtitle, the unified modeling framework for cyber-resilient IT--OT systems that captures how attacks on IT, physical faults on OT, and safety responses during faults jointly shape system risk over time. \frm\ represents possible attacks in cyber space as an attack graph (AG), causal escalation of physical or operational failures as a fault tree (FT), and time-ordered safety barriers and operator or automation decisions that unfold after a fault as an event tree (ET). By explicitly capturing the cascade path, AG $\rightarrow$ FT $\rightarrow$ ET $\rightarrow$ AG $\rightarrow$ FT $\rightarrow \cdots$ (cyber-to-physical triggering followed by safety-response induced exposure that can create new cyber attack surfaces), \frm\ models how benign-seeming safety events or false alarms can enable new OT-to-IT and IT-to-OT pathways. %This integrated framework supports planning-based analysis to automatically discover cascading attack--fault--safety event paths, identify critical barriers and exposures, and reason about system-level outcomes under dynamic IT--OT interactions.

\frm\ considers an IT--OT deployment consisting of (i) a collection of hosts or devices arranged in a specified network architecture (HMI, workstations and PLCs), (ii) a set of physical/operational faults (sensor failure, open valve etc.) that may occur and propagate in the OT process (sensor failure  $\to$ reactor shutdown), and (iii) a time-ordered sequence of safety barriers and operational events (alarm visible, PLC code checks) that unfold once a fault occurs. We model all three layers using propositional attributes. Intuitively, a \emph{host-level} attribute captures cyber or device state (e.g., $\mathsf{has\_vuln}(host,\textsf{CVEs})$, $\mathsf{configurations}(host)$, $\mathsf{reachable}(host_1,host_2)$); a \emph{fault-level} attribute captures OT fault state (e.g., $\mathsf{has\_instrumental\_ failure(host)}$, $\mathsf{has\_ sensor\_ failure(host)}$), and an \emph{event-level} attribute captures safety-response state and temporal stage (e.g., $\mathsf{has\_system\_checked, has\_alarm\_raised, has\_repaired}$). These three sets comprise \frm\ attributes.
%Considering the running example, a IT-OT connected network where components such as HMI, PLC, WorkStations, TEP interact precisely by sharing propositions across layers such that an attack may establish a fault proposition such as `\textit{a mudbus injection on PLC leads to pressure sensor failure}'. The fault proposition i.e., `\textit{pressure sensor failure}' may enable event-tree progression -- `\textit{safety events for pressure sensor failure}'. And an event-tree barrier outcome i.e., `\textit{barrier: check plc logic in workstation}', may establish new host-level exposure propositions i.e, `\textit{barrier-outcome: workstation is live}', that enable subsequent attacks -- `\textit{compromised workstation}'.
%\frm\ consists of a collection of hosts or systems arranged in a specified IT-OT network architecture, each of which is identified by some network label propositions, a set of faults that could occur in the OT captured fault level propositions and a sequence of safety events during an occurrence of faults captures as safety/event level propositions. 
%The  host attribute sets consist of propositional facts that are used to capture network host level information, such as host configurations, vulnerabilities, open ports. The fault level attribute sets capture all the initial, intermediate and top level faults. The event level attributes are the propositions of safety events during faults.

\begin{defn}
Let $\mathbf{H}$ represent the set of unique hosts and their labels, where $\mathbf{H} = \{\mathbf{H}_1, \mathbf{H}_2, \ldots, \mathbf{H}_n\}$ and $P_{\mathbf{H}}$ represent the attribute set of propositions, used to describe the system state in IT-OT network.
Each host $\mathbf{H}_i \in \mathbf{H}$ has a unique collection of attributes, namely, $P_i$, that are true for $\mathbf{H_i}$, where $P_i \subseteq P_\mathbf{H}$ represents a set of propositions about host $\mathbf{H_i}$'s attribute (network-level, fault-level, safety/event level). 
\end{defn}

%\frm\ reasoning operates over the joint valuation space of that AG-relevant propositions (cyber/host facts), FT-relevant propositions (fault facts), and ET-relevant propositions (safety-response facts), represented as $P_{\mathbf{H}}$.
For any instance, an \frm\ system state is a set $s\subseteq P_{\mathbf{H}}$ of propositions that are currently true. Note that attack graph, fault tree and event trees revolve around a finite set of hosts/devices where some propositions can be true for all, so rather than separate proposition set, we define $\mathbf{P_H}$ as global proposition set over all the hosts of IT-OT network. 

\subsection{Attack Graph (AG)}
\label{subsec:ag}
We represent an attack using preconditions and postconditions.
Let $\mathcal{A}$ be the set of attacks associated with a network, where each attack $a \in \mathcal{A}$ is captured by a tuple of the form $a = \langle pre(a), post^+(a),$ $ post^-(a) \rangle $, where for any attack tuple, $pre(a) \subseteq P_\mathbf{H}$ captures the preconditions, $post^+(a)\subseteq P_\mathbf{H}$ captures the positive postconditions (i.e., the set of attributes turned true), and $post^-(a)\subseteq P_\mathbf{H}$ captures the negative postconditions (i.e., the set of attributes turned false). Thus, to capture all the attacks on the network we formally define attack graph (AG) as an AND-OR graph \cite{podder2025spear}.  

\begin{defn}
 An Attack Graph (AG) is a graph,  $\mathbb{G} = \langle\mathbf{N_A}, \mathbf{E_A} \rangle$, where $\mathbf{N_A}$ are the nodes and $\mathbf{E_A}$ are the edges. Nodes represent hosts along with their attributes, namely, $\mathbf{N_A} \subseteq \mathbf{H} \times 2^{P_{\mathbf{H}}}$, and edges are attack edges. 
  An Attack Edge $e \in \mathbf{E}_\mathbf{A}$ of the form $\src(e) \subseteq \mathbf{N_A}$ and $\dst(e)  \subseteq \mathbf{N_A}$ exists iff there is an  attack $a_i$ =  $\langle \pratt{a}{i}, \postp{a}{i},$ $ \postn{a}{i}\rangle$, where
 \begin{itemize}
     \item  $\exists 
        n_i \in \src(e)$,  that corresponds to the same host as the one for $n_i' \in \dst(e)$, where $n_i = (h_i, P_i)$
     \item $\pratt{a}{i} \subseteq \bigcup_{n_j \in \src(e)} P_j$, where $n_j = (h_j, P_j)$
     \item for the node $n_i' = (h_i, P_i')$, we have $P_i' = (P_i \setminus \postn{a}{i}) \cup \postp{a}{i}$ 
 \end{itemize}
\end{defn}

For example, consider an attack $a =$ `\textit{discover PLC via arp spoofing}', the $pre(a)$ is `\textit{attacker present}' $\wedge$ `\textit{HMI is compromised}' $\wedge$ `\textit{PLC is live}' $\wedge$ `\textit{attacker performs ARP spoofing}'. %The postcondition of the attack is `\textit{PLC is discovered}' which becomes precondition of another attack. 
Each attack enables an edge ($e$) from $src(e) = \{attacker, HMI, PLC\}$ to $dst(e) = \{PLC\}$, thus forming attack graph. 
%We assume a monotonic attack model, i.e., attacks do not delete host-level attributes, therefore for every attack edge $e$ we have $P_{\mathbf{H}} \cap post^{-}(a)=\emptyset$, meaning once an attribute becomes true it remains true and the attacker may revisit previously reached nodes.

\subsection{Fault Tree (FT)}
\label{subsec:ft}

We represent a fault occurrence and escalation using preconditions and postconditions over fault-level attributes. Let $\mathcal{F}$ be the finite set of faults in the OT system (including basic, intermediate,
and top-level faults). Each fault $f \in \mathcal{F}$ is captured by a tuple of the form $f = \langle pre(f), post^+(f), post^-(f) \rangle$, where $pre(f) \subseteq P_{\mathbf{H}}$ captures the enabling conditions for $f$ to occur
(e.g., the occurrence of basic faults, intermediate faults, or operating conditions), $post^+(f)\subseteq P_{\mathbf{H}}$ captures the positive postconditions (fault attributes that become true after $f$ occurs), and $post^-(f)\subseteq P_{\mathbf{H}}$ captures the negative postconditions (fault attributes that become false after $f$ occurs, if any). We treat the fault tree (with AND/OR gates) as inducing a transition system over fault propositions,
yielding the directed tree $\mathbb{F}$ defined below.
% Please note that for 

\begin{defn}
A Fault Tree (FT) is denoted as $\mathbb{F}=\langle \mathbf{N_F},\mathbf{E_F}\rangle$, where $\mathbf{N_F}$ are nodes and $\mathbf{E_F}$ are directed fault edges.
Nodes represent hosts with fault-level attributes, $\mathbf{N_F} \subseteq \mathbf{H}\times 2^{P_{\mathbf{H}}}$.

For any edge $e\in \mathbf{E_F}$, let $\src(e)\subseteq \mathbf{N_F}$ denote its source node-set and let $\dst(e)\in \mathbf{N_F}$ denote its destination node.
A fault edge $e$ with $\src(e)$ and $\dst(e)$ exists iff there exists $f_j = \langle pre^j(f),post^{+j}(f),post^{-j}(f)\rangle$ fault such that:
\begin{itemize}
    \item $\exists\, n_j\in \src(e)$ and $n_j'\!=\!(h_j,P_j')$ correspond to the same host $h_j$, where $n_j=(h_j,P_j)$
    \item $pre^j(f)\subseteq \bigcup_{n_k\in \src(e)} P_k$, where $n_k=(h_k,P_k)$;
    \item $P_j'=(P_j\setminus post^-(f))\cup post^+(f)$ for node $n_j' = (h_j, P_j')$.
\end{itemize}
\end{defn}

Note that we treat the fault tree as a host-scoped fault-propagation tree that encodes how fault conditions can occur and escalate at specific systems in the IT–OT deployment. Concretely, each fault proposition is interpreted relative to a host, instead of modeling a fault as a standalone event, we model `\textit{fault-level propositions $p_f$ holds at host $h$}' using a host-indexed proposition $(h,p_f)$ where $h\in \mathbf{H}$ and $p_f\subseteq P_{\mathbf{H}}$. For example, `\textit{pressure sensor failure}' is a fault-level proposition $p_i$ that holds at host $h_i$, `\textit{Simulated TE Process plant}'. Together,
these define the fault node \(n_i=(h_i,p_i)\).

%A fault ($f \in \mathcal{F}$) is represented using an action-like tuple (preconditions and postconditions). Intuitively, when the enabling conditions of a fault become true (possibly requiring cyber facts, physical operating conditions, or previously established fault facts), the fault ``fires'' and its postconditions mark new fault propositions as true on the corresponding host. The directed edges of the fault tree capture exactly these host-level fault transitions, from a source node-set whose propositions satisfy the fault’s preconditions to a destination node that reflects the updated fault state for the affected host.
%Also, we assume a monotonic fault-propagation model, i.e., fault transitions do not delete fault-level attributes; therefore for every fault tuple $e$ we have $P_{\mathbf{H}} \cap post^{-}(f)=\emptyset$, meaning once a fault condition becomes true it remains true. 

\subsection{Event Tree (ET)}
\label{subsec:et}

An Event Tree (ET) is a forward, consequence-oriented model that describes how a system evolves after an initiating event (a fault) by enumerating a time-ordered sequence of safety barriers such as operator actions, procedures, and automatic protections. Each barrier yields a discrete outcome (i.e., yes or no) along with a change of states in the host, and the combination of barrier outcomes determines the resulting operational or safety states. %In \frm, the event tree captures how safety responses to faults can both mitigate escalation and, in some cases, unintentionally expose new cyber attack surfaces that enable further attack and cascading failures.

\begin{defn}
    An initiating event ($\mathbf{i_e}$) is the first significant deviation from a normal operating condition that may lead to unwanted consequences. 
\end{defn}

Note that by definition each event tree is triggered by exactly one initiating event ($i_e$) but as we are discussing about IT-OT network, there could be multiple faults and multiple initial events ($I_E$) which leads to multiple ETs for the network. To reduce notation overload all the below discussion is for single event tree initiated due to an initial event ($i_e$).

A timed barrier system in event tree (ET) can be modeled as an ordered collection of barriers together with discrete time steps and a binary outcome condition for each barrier. 
\begin{defn}
    Let $\mathbf{B}=\{ b_1,\dots,b_m\}$ be the ordered set of barriers and let $\mathbf{T}=\{ t_0,t_1,\dots,t_m\}$ be the ordered set of discrete time steps for initial event $i_e$, encoded as propositions $\mathsf{at}(t)\in P_{\mathbf{H}}$. For each barrier $b_i\in\mathbf{B}$, there is a propositional condition $\mu(b_i)\in P_{\mathbf{H}}$ that determines the barrier outcome (true or false).
\end{defn}
%  We represent each barrier $b_i \in \mathbf{B}$ using a precondition set $pre(b_i)\subseteq P_{\mathbf{H}}$ and outcome-dependent postconditions $(post^{+T}(b_i),$ $post^{-T}(b_i))$ for the true outcome and $(post^{+F}(b_i),post^{-F}(b_i))$ for the false outcome, where all pre- and post-conditions are subsets of $P_{\mathbf{H}}$. The time-ordered evaluation of barriers is captured by including the stage predicate $\mathsf{at}(t_{i-1})$ in $pre(b_i)$. Regardless of outcome, barrier evaluation advances the stage by ensuring $\mathsf{at}(t_i)$ is added and $\mathsf{at}(t_{i-1})$ is removed. Thus, for any event-state $s\subseteq P_{\mathbf{H}}$ satisfying $pre(b_i)\subseteq s$, the successor state $s'$ satisfies
% $\mathsf{at}(t_i)\in s'$ and $\mathsf{at}(t_{i-1})\notin s'$, and the remaining changes are determined by the outcome of $\mu(b_i)$ such as if $\mu(b_i)\in s$ then $post^{+T}(b_i)\subseteq s'$ and $post^{-T}(b_i)\cap s'=\emptyset$, whereas if $\mu(b_i)\notin s$ then $post^{+F}(b_i)\subseteq s'$ and $post^{-F}(b_i)\cap s'=\emptyset$.
% Finally, each barrier is marked as completed using an execution proposition $\mathsf{resolved}(b_i)\in P_{\mathbf{H}}$ that is included in both $post^{+T}(b_i)$ and $post^{+F}(b_i)$, ensuring that each $b_i$ is evaluated exactly once while its true/false outcome selects between the two postcondition sets.
The deterministic barrier model reflects regulatory practice in ICS safety engineering. Under IEC 61511 and ISA-84, each barrier outcome is prescribed in advance by domain experts for each initiating event — the safety procedure is a fixed document, not a runtime decision. Accordingly, barrier truth values $\mu(b_i)$ can be changed based on operator's decision as an operator has a choice to implement the barrier or not.

\begin{defn}
\label{def:et}
An Event Tree (ET) is defined as a tuple, $\mathbb{E}=\langle \mathbf{i_e}, \mathbf{B}, \mu, \mathbf{T}, \mathbf{O}\rangle$ where:
\begin{itemize}
    \item $\mathbf{i_e}$ is the initiating-event propositions that triggers the ET.
    \item $\mathbf{B}=\{ b_1,\dots,b_m\}$ is an ordered set of timed barriers.
    \item $\mathbf{T}=\{ t_0,t_1,\dots,t_m\}$ is an ordered set of discrete time steps represented by propositions $\mathsf{at}(t)\in P_{\mathbf{H}}$.
    \item $\mu:\mathbf{B}\rightarrow P_{\mathbf{H}}$ assigns each barrier $b_i$ a propositional condition $\mu(b_i)$ whose truth value determines the barrier outcome (True/False).
    \item Each barrier $b_i$ is associated with a common precondition set $pre(b_i)\subseteq P_{\mathbf{H}}$ and two outcome-dependent postcondition pairs $(post^{+T}(b_i),post^{-T}(b_i))$ and $(post^{+F}(b_i),post^{-F}(b_i))$, all subsets of $P_{\mathbf{H}}$, where $\mathsf{at}(t_{i-1})\in pre(b_i)$ and $\mathsf{at}(t_i)\in post^{+T}(b_i)\cap post^{+F}(b_i)$ and $\mathsf{at}(t_{i-1})\in post^{-T}(b_i)\cap post^{-F}(b_i)$.
    \item $\mathbf{O}\subseteq P_{\mathbf{H}}$ is the set of ET outcome propositions.
\end{itemize}
For any event state $s\subseteq P_{\mathbf{H}}$ satisfying $i_e\in s$ and $pre(b_i)\subseteq s$, the successor state $s'$ after evaluating $b_i$ is defined by:
\begin{align*}
    s'=
\begin{cases}
\big(s\setminus post^{-T}(b_i)\big)\cup post^{+T}(b_i) & \text{if }\mu(b_i)\in s,\\[2pt]
\big(s\setminus post^{-F}(b_i)\big)\cup post^{+F}(b_i) & \text{if }\mu(b_i)\notin s.
\end{cases}
\end{align*}
\end{defn}
A complete ET execution for $i_e$ is a sequence of evaluations $\{ b_1,\dots,b_m\}$ starting at $\mathsf{at}(t_0)$ and ending at $\mathsf{at}(t_m)$, producing an outcome valuation over $\mathbf{O}$ determined by the barriers.% outcomes.

As described in Algorithm \ref{alg:et-timed-yesno}, after executing the ET for initiating event $i_e$, the evaluation of each barrier $b_i$ yields a binary result determined by whether $\mu(b_i)$ holds in the state at the time of evaluation. We capture these barrier results as an outcome vector $\omega=\langle \omega_1,\dots,\omega_m\rangle$, where each component records the observed outcome of barrier $b_i$, i.e., $\omega_i=T$ if $\mu(b_i)$ holds (equivalently, $True(b_i)$ becomes true), and $\omega_i=F$ otherwise (equivalently, $False(b_i)$ becomes true). The overall event-tree outcome is defined by an outcome function $\mathsf{out}:\{T,F\}^m \rightarrow \mathbf{O}$,
(or more generally $\mathsf{out}:\{T,F\}^m \rightarrow 2^{\mathbf{O}}$ when multiple outcome propositions may hold simultaneously), which maps a complete barrier-outcome pattern to the corresponding ET consequence class. Intuitively, $\mathsf{out}$ encodes the ET leaf semantics such that each root-to-leaf path corresponds to a unique truth assignment to the barriers, and $\mathsf{out}(\omega)$ returns the outcome proposition(s) labeling that leaf (e.g., \textit{mitigated}, \textit{degraded}, \textit{unsafe-escalation}, \textit{exposure-created}). Therefore, once $\omega$ is obtained after evaluating $\{ b_1,\dots,b_m\}$, the ET outcome is given by $O=\mathsf{out}(\omega).$

\begin{algorithm}[tbp!]
\scriptsize
\caption{Timed Event-Tree Execution with Yes/No Barriers}
\label{alg:et-timed-yesno}
\vspace{2pt}
\SetKwInOut{Input}{Input}
\SetKwInOut{Output}{Output}
\Input{Initiating event $i_e\in \mathbf{I_E}$, ordered barriers $\mathbf{B}=\{ b_1,\dots,b_m\}$, time steps $\mathbf{T}=\{ t_0,\dots,t_m\}$, condition map $\mu:\mathbf{B}\rightarrow P_{\mathbf{H}}$, barrier specs $pre(\cdot),post^{+T}(\cdot),post^{-T}(\cdot),post^{+F}(\cdot),post^{-F}(\cdot)$, outcome map $\mathsf{out}:\{\mathsf{True},\mathsf{False}\}^m\rightarrow \mathbf{O}$}
\Output{Outcome $O$ and final event-state $s_{\mathbf{H}}$}
\vspace{2pt}

$s_{\mathbf{H}} \leftarrow \{i_e,\mathsf{at}(t_0)\}$\;
$\omega \leftarrow \langle\ \rangle$\;

\For{$i \leftarrow 1$ \KwTo $m$}{
    \tcp{Timed enablement and ordering}
    \If{$pre(b_i)\not\subseteq s_{\mathbf{H}}$}{
        \textbf{return} $(\langle\ \rangle,\; s_{\mathbf{H}})$\;
    }

    \tcp{Evaluate barrier outcome \& update $s_{\mathbf{H}}$}
    \eIf{$\mu(b_i)\in s_{\mathbf{H}}$}{
        $\omega_i \leftarrow \mathsf{True}$\;
        $s_{\mathbf{H}} \leftarrow \big(s_{\mathbf{H}} \setminus post^{-T}(b_i)\big)\ \cup\ post^{+T}(b_i)$\;
    }{
        $\omega_i \leftarrow \mathsf{False}$\;
        $s_{\mathbf{H}} \leftarrow \big(s_{\mathbf{H}} \setminus post^{-F}(b_i)\big)\ \cup\ post^{+F}(b_i)$\;
    }
    Append $\omega_i$ to $\omega$\;
}

$O \leftarrow \mathsf{out}(\omega)$\;
\textbf{return} $(O,\; s_{\mathbf{H}})$\;
\end{algorithm}

\subsection{\frmtitle}
\label{subsec:sg}
%In this work we treat the Event Tree (ET) as a \emph{structured object}(initiating event, ordered barriers, and an outcome function) rather than as an explicit graph. Consequently, \frm\ does not introduce separate ``event nodes'' or ``event edges''.

\frmtitle\ unifies \darkcircled{1} the attack graph $\mathbb{G}=\langle \mathbf{N_A},\mathbf{E_A}\rangle$,
\darkcircled{2} the fault tree $\mathbb{F}=\langle \mathbf{N_F},\mathbf{E_F}\rangle$, and
\darkcircled{3} a family of event-tree objects $\{\mathbb{E}_{i_e}\}_{i_e\in \mathbf{I_E}}$,
one ET for each initiating-event proposition, where $\mathbf{I_E}$ is set of all possible initial event $i_e$.
%For each $i_e\in \mathbf{I_E}$, $\mathbb{E}_{i_e}=\langle i_e,\mathbf{B}^{i_e},\mu^{i_e},\mathbf{T}^{i_e},\mathbf{O}^{i_e}\rangle$ is activated when $i_e$ becomes true, and specifies an ordered barrier sequence $\mathbf{B}=\{ b_1,\ldots,b_m\}$ evaluated over discrete time steps $\mathbf{T}=\{ t_0,\ldots,t_m\}$. 
For ease of representation, $\mathbb{E}_{i_e}=\langle i_e,\mathbf{B}^{i_e},\mu^{i_e},\mathbf{T}^{i_e},\mathbf{O}^{i_e}\rangle$ will be subsequently denoted as $\mathbb{E}=\langle \mathbf{i_e}, \mathbf{B}, \mu, \mathbf{T}, \mathbf{O}\rangle$.
%are the same event tree for $i_e$. We do not represent the former way due to notation overload.
%Barrier evaluation yields a binary outcome vector $\omega\in\{T,F\}^m$, and the corresponding ET consequence is $O=out(\omega)\in\mathbf{O}$. %Please note here $\mathbf{B},\mu,\mathbf{T},\mathbf{O}$ corresponds to single initial event $i_e$ not for all 

Cross-domain dependencies in \frm\ are captured \emph{without} separate linking relations. %All components share the same propositional set $P_{\mathbf{H}}$ of host's attribute. 
Cyber-to-physical triggering is represented by faults whose enabling conditions include
host-level propositions established by attacks (e.g., $pre(f)\subseteq P_{\mathbf{H}}$ may contain propositions produced by some attack action). Fault-to-event activation is represented by initiating-event propositions $i_e\in \mathbf{I_E}\subseteq P_{\mathbf{H}}$
that are produced by fault occurrences. Safety-response induced cyber exposure is represented by barrier effects and/or ET outcomes which add host-level propositions in $P_{\mathbf{H}}$ that can serve as preconditions of subsequent attacks. We summarize this exposure using a mapping $\psi:\mathbf{O}\rightarrow 2^{P_{\mathbf{H}}}$, where $\psi(O)$ denotes the set of propositions made available after ET outcome $O$. Now we can formally define the \frmtitle\ framework.

\begin{defn}
\frmtitle\ is defined as the unified model
\[
\frm\ =\langle \mathbb{G},\mathbb{F},\bigcup_{i_e\in \mathbf{I_E}} \{\mathbb{E}_{i_e}\}\rangle,
\]

where $\mathbb{G}=\langle \mathbf{N_A},\mathbf{E_A}\rangle$ is an attack graph, $\mathbb{F}=\langle \mathbf{N_F},\mathbf{E_F}\rangle$ is a fault tree, and for each $i_e\in \mathbf{I_E}$,
$\mathbb{E}_{i_e}=\langle i_e,\mathbf{B}^{i_e},\mu^{i_e},\mathbf{T}^{i_e},\mathbf{O}^{i_e}\rangle$
be an event-tree triggered by $i_e$.

\end{defn}

Please note that the definitions above — $\mathbb{G},\mathbb{F}, \mathbb{E}$, and their unification as $SG$ — constitute the complete modeling framework. %For formal completeness, Appendix \ref{app:extra-def} defines derived path notions (attack paths, fault paths, event paths, and their cross-layer compositions) that characterize the structural types of traces \frm\ can represent. These definitions follow directly from the framework above and are not required to understand the planning compilation or the case study analyses, thus they are provided as a reference for extending the framework.
%All analysis performed on \frm\  will be represented in terms of presence or absence of path that exists through $\mathbb{G},\mathbb{F}, \mathbb{E}$ of IT-OT network. 
However, instead of performing the analysis directly on  (\frm) graph, we first compile them into a planning problem and use planners to perform the analysis. Since planning scales to large problem spaces, it is well-suited for automated reasoning over complex IT-OT deployments~\cite{podder2025spear, bashir2024resiliency}. 

\subsection{Planning Compilation for \frm}
\label{subsec:planning}

A classical planning problem is defined by a tuple $\mathcal{M} = \langle F, A, I, G \rangle$, where $F$ is a finite 
set of propositional fluents that characterizes the state space. A state $s \subseteq F$ is the set of fluents currently true. $A$ is the set of actions available to the planner, $I \subseteq F$ is the initial state, and $G \subseteq F$ is the goal state. A planning problem encodes a transition system whose nodes are states and whose directed edges are state transitions induced by action execution. Each action $a \in A$ is defined by 
the tuple $a = \langle \textit{pre}(a), \textit{add}(a), \textit{del}(a) \rangle$, where $\textit{pre}(a) \subseteq F$ specifies the preconditions that must hold for $a$ to be executable, $\textit{add}(a) \subseteq F$ specifies the fluents set to true upon execution, and $\textit{del}(a) \subseteq F$ specifies the fluents set to false. Executing action $a$ in state $s$ produces a successor state via the transition function,  $\Gamma_{\mathcal{M}}(s, a) = (s \setminus \textit{del}(a)) \cup \textit{add}(a), 
\quad \text{if } \textit{pre}(a) \subseteq s$. 
%$\Gamma_{\mathcal{M}}$:
A solution to a planning problem, called a \emph{plan}, is a sequence of actions $\pi = \langle a_1, \ldots, a_k \rangle$ such that $\Gamma_{\mathcal{M}}(I, \pi) \supseteq G$, where 
$\Gamma_{\mathcal{M}}$ is extended to action sequences in the natural way. A plan is \emph{optimal} if no shorter valid plan exists.

To compile a \frm\ instance into a classical planning model $\mathcal{M}=\langle F,A,I,G\rangle$, we treat every propositional fact that can become true during a valid path as a planning fluent, and we treat every possible evolution step (attack, fault propagation, and barrier evaluation) as a planning action. The fluent set $F$ is formed from a \emph{global attribute set of propositional facts} such that host/device
facts are represented as host-indexed fluents for pairs $(h,p)\in \mathbf{H}\times P_{\mathbf{H}}$ using a mapping function $\gamma^{\mathcal{M}}$, while event-tree
facts (e.g., stage propositions $\mathsf{at}(t)$, barrier-status propositions $\mathsf{resolved}(b)$, decision
propositions $\mu(b)$, and outcome propositions) are represented as \emph{global} fluents using a mapping function $\varphi^{\mathcal{M}}$ since the ET is modeled as a structured object over propositions rather than as host-level state.

%The action set $A$ is the union of three action families: (i) for each attack edge in $\mathbf{E_A}$, we introduce an action whose preconditions correspond to the attributes labeling the edge's source node-set, and whose effects apply the associated attack tuple's $post^{+}$ and $post^{-}$ update to the destination host; (ii) for each fault edge in $\mathbf{E_F}$, we introduce an action defined analogously using the fault tuple's $pre(\cdot)$ and $post^{\pm}(\cdot)$; (iii) for each barrier $b$ in event tree $\mathbb{E}_{i_e}$, we introduce a single barrier-evaluation action that advances the ET stage and marks $b$ as resolved, while selecting the True versus False branch using the decision proposition $\mu(b)$.The truth values of decision propositions $\mu(b)$ are fixed in $I$, reflecting the deterministic safety procedures specified in domain expert documentation (see $\S$~\ref{subsec:et}). Barrier branching is thus compiled without introducing nondeterminism.
%keeping $\mathcal{M}$ within positive-precondition STRIPS.
%The initial state $I$ consists of the host facts induced by the given initial node-set $S_N^0$ together with the ET initialization facts (e.g., $\mathsf{at}(t_0)$) and the fixed decision propositions $\mu(b)$. Finally, the goal $G$ is typically a target host-attribute fluent (e.g., achieving $(h_t,p_t)$), but it also encodes fault-level or safety-level objectives such as reaching a top-fault proposition or an unsafe event-tree outcome.

\begin{defn}
\label{def:planning}
For a given \frm, $S_N^0$ be an initial node-set and let $(h_t,p_t)$ be the target host-attribute, a representative classical planning model is defined as
$\mathcal{M}^{\mathsf{SG}}=\langle F^{\mathsf{SG}},A^{\mathsf{SG}},I^{\mathsf{SG}},G^{\mathsf{SG}}\rangle$,
where:

\begin{itemize}
    \item
    $F^{\mathsf{SG}}$ contains fluents for host-attribute pairs in $\mathbf{H}\times P_{\mathbf{H}}$ and or event-tree propositions relevant to ET execution.
    $\gamma^{\mathcal{M}}$ maps $\mathbf{H}\times P_{\mathbf{H}}$ to $F^{\mathsf{SG}}$. For each ET proposition $p \in \mathbf{P_H}$ relevant to ET execution (for example $q \in \{at(t), resolved(b), \mu(b), O, i_e\}$ is mapped to  $F^{\mathsf{SG}}$ using $\varphi^{\mathcal{M}}$ (to simplify notation we will also allow the application of the function over sets). 

    \item
    $A^{\mathsf{SG}}$ contains three families of actions:
    \begin{enumerate}
        \item
        For each attack edge $e\in \mathbf{E_A}$ with destination node $(h_i,P_i)$, create an action $a_e$ where
        \[
        pre(a_e)=\bigcup_{(h,P)\in src(e)} \gamma^{\mathcal{M}}(\{(h,p)\mid p\in P\}).
        \]
        Let $a=\langle pre(a),post^+(a),post^-(a)\rangle\in \mathcal{A}$ be the attack tuple inducing $e$.
        Then $add(a_e)=\gamma^{\mathcal{M}}(\{h_i\}\times post^+(a))$, $del(a_e)=\gamma^{\mathcal{M}}(\{h_i\}\times post^-(a))$.

        \item
        For each fault edge $e\in \mathbf{E_F}$ with destination node $(h_j,P_j)$, create an action $a_e$ where
        \[
        pre(a_e)=\bigcup_{(h,P)\in src(e)} \gamma^{\mathcal{M}}(\{(h,p)\mid p\in P\}).
        \]
        Let $f=\langle pre(f),post^+(f),post^-(f)\rangle\in \mathcal{F}$ be the fault tuple inducing $e$.
        Then
        $add(a_e)=\gamma^{\mathcal{M}}(\{h_j\}\times post^+(f))$,
        $del(a_e)=\gamma^{\mathcal{M}}(\{h_j\}\times post^-(f))$.
        %(With our monotonic-fault assumption, $post^-(f)=\emptyset$ and thus $del(a_e)=\emptyset$.)

        \item 
        For each event tree $\mathbb{E}_{i_e}\in \mathbf{E_\mathbb{I_E}}$ with ordered barriers
        $\mathbf{B}=\langle b_1,\dots,b_m\rangle$ and time steps $\mathbf{T}=\langle t_0,\dots,t_m\rangle$,
        include action $a_{b_k}$ for each barrier $b_k$.
        The precondition is
        \[
        pre(a_{b_k})=\varphi^{\mathcal{M}}\{(p)\mid p\in pre(b_k)\})
        \]
        where $pre(b_k)$ includes the required stage proposition $\mathsf{at}(t_{k-1})$ and any enabling facts
        (including $i_e$ when needed).
        The effects contain conditional branch determined by the decision proposition $\mu(b_k)\in P_{\mathbf{H}}$,\\
            $\textbf{if } \varphi^{\mathcal{M}}((\mu(b_k))) \textbf{ then }$
            $add(a_{b_k})=\varphi^{\mathcal{M}}(\{p\mid p\in post^{+T}(b_k)\})$,
            $del(a_{b_k})=\varphi^{\mathcal{M}}(\{(h_\star,p)\mid p\in post^{-T}(b_k)\})$,
            \textbf{else }
            $add(a_{b_k})=\varphi^{\mathcal{M}}(\{p\mid p\in post^{+F}(b_k)\})$,
            $del(a_{b_k})=\varphi^{\mathcal{M}}(\{p\mid p\in post^{-F}(b_k)\})$.
        Here the truth value of $\mu(b_k)$ is fixed by the initial state (operator/automation decision as an
        initial condition). Please note that $\varphi^{\mathcal{M}}(\mathsf{at}(t_k))),\varphi^{\mathcal{M}}(\mathsf{resolved}(b_k))$ is part of $add(a_{b_k})$ and $del(a_{b_k})$ for both $\varphi^{\mathcal{M}}((\mu(b_k)))$ true or false.
    \end{enumerate}

    \item
    $I^{\mathsf{SG}}=\Big(\bigcup_{(h,P)\in S_N^0}\gamma^{\mathcal M}(\{(h,p)\mid p\in P\})\Big)
    \ \cup\ \varphi^{\mathcal M}(\{\mathsf{at}(t_0)\}\cup\{\mu(b)\mid b\in\mathbf{B}\}).$

    \item
    $G^{\mathsf{SG}}=\{\gamma^{\mathcal{M}}((h_t,p_t))\}$.
\end{itemize}
\end{defn}

We now show that the representative planning model $\mathcal{M}^{SG}$ is a sound and complete representation of the \frm\ instance.

\begin{theorem}
\label{thm:soundness-completeness}
Let $SG = \langle \mathbb{G}, \mathbb{F}, \bigcup_{i_e \in I_E} \{\mathbb{E}_{i_e}\} \rangle$ be an  \frmtitle\ instance, $S^0_N$ be an initial node-set, and let $(h_t, p_t)$ be a target host-attribute. The planning model $\mathcal{M}^{SG} = \langle F^{SG}, A^{SG}, I^{SG}, G^{SG} \rangle$ is a sound and complete representation of \frm:
\begin{itemize}
    \item[\textbf{C1}] \emph{(Soundness)} Any valid plan $\pi$ in $\mathcal{M}^{SG}$ corresponds to a valid cross-layer path through \frm, i.e., a sequence of attack edges in $\mathbb{G}$, fault edges in $\mathbb{F}$, and barrier evaluations in $\{\mathbb{E}_{i_e}\}_{i_e \in I_E}$.% that reaches the target proposition $\gamma^\mathcal{M}(h_t, p_t)$ from the initial node-set $S^0_N$.
    \item[\textbf{C2}] \emph{(Completeness)} Any valid path through \frm\ from $S^0_N$ to $(h_t, p_t)$ corresponds to a valid plan $\pi$ in $\mathcal{M}^{SG}$.
\end{itemize}
\end{theorem}

\begin{proof}[Proof Sketch]
The proof follows from the observation that the compilation in Definition~\ref{def:planning} establishes a one-to-one correspondence between edges in \frm\ and actions in $\mathcal{M}^{SG}$.
By construction, $A^{SG}$ contains exactly three families of actions. (i)~For each attack edge $e \in E_A$, the corresponding action $a_e$ has preconditions 
$\gamma^\mathcal{M}(\bigcup_{(h,P) \in src(e)} \{(h,p) \mid p \in P\})$ and effects $\gamma^\mathcal{M}(\{h_i\} \times post^+(a))$ and $\gamma^\mathcal{M}(\{h_i\} \times post^-(a))$, mirroring exactly the precondition and postcondition structure of the attack tuple $a$ inducing $e$. (ii)~For each fault edge $e \in E_F$, the corresponding action $a_e$ is defined analogously from the fault tuple $f$ inducing $e$. (iii)~For each barrier $b_k$ in each event tree $E_{ie}$, the corresponding action $a_{b_k}$ requires 
the stage proposition $\varphi^\mathcal{M}(at(t_{k-1}))$ and produces $\varphi^\mathcal{M}(at(t_k))$, advancing the ET execution by exactly one step and applying the True or False branch postconditions determined by 
the fixed truth value of $\varphi^\mathcal{M}(\mu(b_k))$ in the initial state $I^{SG}$.

\textbf{Soundness (C1).} Let $\pi = \langle a_1, \ldots, a_n \rangle$ be a valid plan in $\mathcal{M}^{SG}$, so that $\Gamma_{\mathcal{M}^{SG}}(I^{SG}, \pi) \supseteq G^{SG}$. Each action $a_i \in \pi$ belongs to exactly one of the three families above. An action is executable at a state $s$ only if its preconditions are satisfied in $s$; by the correspondence above, this means the source node-set of the equivalent edge in $G^{SG}$ or $F^{SG}$ is present in the current propositional state, or the required ET stage and enabling propositions hold for the barrier action. Executing $a_i$ updates the state by exactly the add and delete effects of the corresponding edge or barrier evaluation. Therefore, $\pi$ 
induces a valid sequence of attack edges, fault edges, and barrier evaluations in \frm\ that reaches a state containing $\gamma^\mathcal{M}(h_t, p_t)$, satisfying the target host-attribute $(h_t, p_t)$.

\textbf{Completeness (C2).} Let $\rho$ be a valid cross-layer path in \frm\ from $S^0_N$ to $(h_t, p_t)$, comprising a sequence of attack edges, fault edges, and barrier evaluations. By the action-step correspondence, each edge or barrier evaluation in $\rho$ maps to a unique action in $A^{SG}$ whose preconditions and effects faithfully encode the enabling conditions and postconditions of that step. The initial state $I^{SG}$ encodes the propositions of $S^0_N$ together with the ET initialization facts $\varphi^\mathcal{M}(at(t_0))$ and the fixed decision propositions $\varphi^\mathcal{M}(\mu(b))$ for all barriers. Since $\rho$ is valid in \frm, the enabling conditions of each step hold when that step is taken; the corresponding action is therefore executable at the matching planning state. Applying all actions in order yields a final state containing $\gamma^\mathcal{M}(h_t, p_t) = G^{SG}$, so the induced action sequence is a valid plan in $\mathcal{M}^{SG}$.
% \textbf{ET determinism.} Note that barrier evaluation actions do not introduce nondeterminism into the planning model. The truth value of each decision proposition $\mu(b_k)$ is fixed in $I^{RG}$, so each barrier action has a single deterministic outcome. This ensures that the compiled model belongs to the class of positive-precondition STRIPS planning problems, for which plan existence is PSPACE-complete~\cite{bylander1994}. The RG++ cross-layer path existence problem is therefore 
% in PSPACE; membership follows from Theorem~\ref{thm:soundness-completeness} and the PSPACE-completeness of plan existence in this class~\cite{bylander1994}.
\end{proof}

This equivalence also allows us to prove the complexity of finding a cascade path through \frmtitle\ .

\begin{theorem}
For a general Safety Graph, the problem of determining the existence of a cascade path through the \frmtitle\ is PSPACE-complete.
\end{theorem}
\begin{proof}[Proof Sketch]
The membership proof, i.e., the problem is within the class of PSPACE, is given by Theorem 1, as a planning problem of the form discussed here is PSPACE-complete \cite{nebel2000compilability}. The model described above belongs to that class, proving the membership. 

Regarding hardness, we can show that for any planning problem, we can reduce it to a problem of identifying a path for the Safety Graph in polynomial time. To do so, 
we consider a SG with a single host, a single fault tree node (eliminating any corresponding fault edges and fault actions), and a corresponding single event tree $\mathbb{E}_{i_e}$ with a trivial single barrier (containing no conditional effects). 
For this planning problem, given by the tuple $\langle F^{SG},A^{SG},I^{SG},G^{SG}\rangle$, the fluents $F$ map to the attribute set $P_{\mathbf{H}}$ of the SG, the actions $A^{SG}$ map to the attack edges in the attack graph, the initial state $I^{SG}$ maps to the facts from $P_{\mathbf{H}}$ true in the initial node $S_N^0$, and the goal $G^{SG}$ maps to the node in the attack graph which initiates the fault in the single fault tree node. If the corresponding planning problem has a goal containing multiple propositions, a new attack is added with the preconditions being the original goals, and the postcondition is a single new attribute that is the target of the SG.
%With this, there is only a path in the SG if there is a corresponding plan in the original planning problem.
Thus, the constructed \frm\ instance has a cascade path if and only if the original planning problem has a valid plan.
\end{proof}

% ------------------------
% \section{Complexity Analysis}
% \input{41-Complexity-Proof}
% ------------------------
\section{Analysis of \frm\ for Running Example}
\label{sec:analysis}

In this section we demonstrate that \frm\ can discover cross-layer cascade paths that are invisible to models covering only the cyber and physical layers. %Critically, the defender does not know in advance whether any such path exists, nor whether a safety response is the enabling mechanism — both are determined by the planner. 
We instantiate \frm\ on the TEP test-bench, compile it to a planning problem, and pose a question, ``\textit{does a feasible path from an IT foothold to reactor rupture (fault) exist, and if so, what enables it?}'' We answer this in two phases. First, we establish that the analysis of the critical path is unreachable under AG+FT alone. We also run the same planning compilation in RG and ERG \cite{bashir2024resiliency,ray2023explainable}. Second, we show that integrating the safety measures (ET) causes the planner to discover a feasible cascade, and that the returned plan identifies a safety barrier as the precise enabling mechanism. Figure~\ref{fig:workflow} shows the workflow of generation of \frm\ instance and analysis performed using planner.

%We answer this question in two phases. First, we analyze the system using only the AG and FT layers — the standard cyber-physical model — and show the planner cannot reach the critical failure. Second, we integrate the ET layer into the full SG and show the planner discovers a feasible cascade path that traverses AG → FT → ET → AG → FT. The difference between these two results is the paper's central empirical claim.

\subsection{Constructing the SG Instance}
\label{subsec:testbench}

\begin{figure*}[!h]
    \centering
    \includegraphics[width=.7\linewidth]{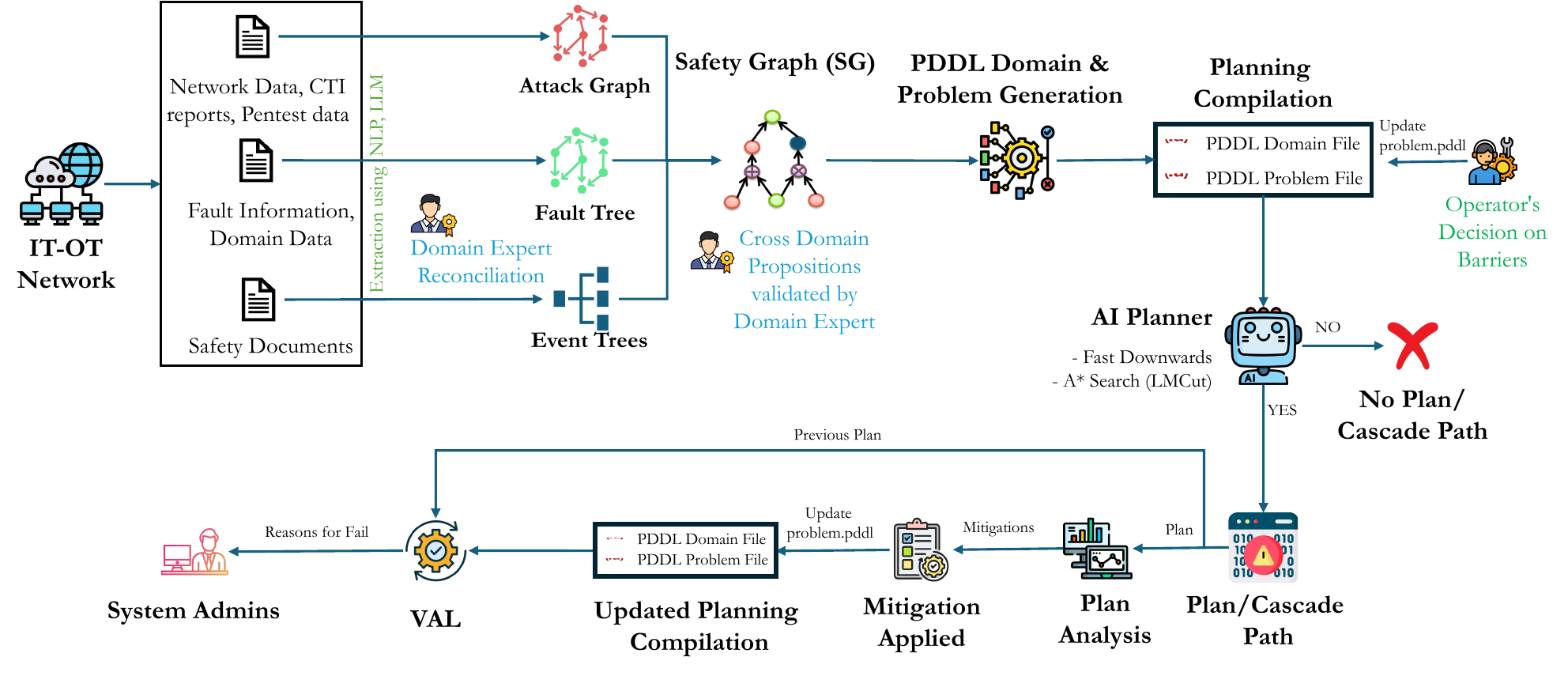}
    \caption{Workflow for modeling \frm\ instance and analyzing \frm\ using AI planner.}
    \label{fig:workflow}
\end{figure*}

%The AG captures the reachable cyber attack surface from an external attacker positioned on the IT subnet. 
For the given network (see $\S$~\ref{sec:runnning}), we collected: (i) network configuration artifacts (host inventory, subnet segmentation, firewall rules), (ii) penetration testing results (open ports, exploitable services, reachable ICS endpoints) \cite{formby2018lowering}, and (iii) CTI reports identifying known vulnerabilities in the deployed software stack to construct the corresponding AG. FT is constructed based on fault analysis report on TEP \cite{zhao2017improved, chiang2001tennessee, formby2018lowering} as a host-scoped fault-propagation structure using Relyence \cite{relyence-fault-tree}. The ET layer is derived from safety manuals following IEC 61511 procedures. For the TEP instantiation, 32 faults were modeled, each with a corresponding event tree of 5--6 barriers, yielding approximately 190 barriers in total — all grounded in certified safety documentation. For the interested reader, the details of the AG, FT and ET for an initiating event `\texttt{pressure-sensor-failure}' are shown respectively in 
Appendix~\ref{app:ag-actions} Figure \ref{fig:ag}, Appendix~\ref{app:ft-actions}, Figures \ref{fig:ft1} and \ref{fig:ft2}, and  
Appendix~\ref{app:et-actions} Figure \ref{fig:et}. We use \textit{llama3:8b} based NER (Named Entity Recognition) methods to extract data from various reports. All the LLM extracted propositions are validated by domain experts.

Given the instantiated AG, FT, and ET, we use automated scripts that generates the AI planner's \texttt{domain.pddl} and \texttt{problem.pddl} files for the corresponding \frm\ instance.
The planning compilation maps host, fault, ET stage/barrier, and outcome propositions into
planning fluents, and maps attack steps, fault transitions, and barrier evaluations into planning actions. The operator's choice to implement barrier or not is captured in the \texttt{problem.pddl} file as barrier's truth value. Once the planning problem is constructed, we use the Fast Downward planner \cite{helmert2006fast} with A* and the LM-cut admissible heuristic to perform analysis. For each variation of operator's decision and goal condition we ran the planner iteratively. Once planner finds a plan or cascade path, we analyze the plan to find the reason why a cascade path exists. Based on the analysis, a mitigation or solution is applied which updates the \texttt{problem.pddl} file. Then the updated compilation and the previous plan are analyzed using VAL \cite{val-planning} to show how the mitigation causes the plan failure.

\subsection{Analysis on the Running Example}
\label{subsec:scenario1}

We first ask whether the critical failure is reachable using only the cyber and physical layers, with no safety measures. This corresponds to the standard cyber-physical threat model used in RG~\cite{bashir2024resiliency} and ERG~\cite{ray2023explainable}. The defender's question is: ``\textit{can an attacker trigger reactor rupture (fault: REX-RUP) failure starting from an IT foothold?}'' The planner finds the partial cascade shown in Figure \ref{fig:fault-plan}. Four attack steps traverse the AG — subnet scan, ARP spoofing, firewall bypass, Modbus injection — and trigger the first fault transition (`\texttt{fault-reactor-pressure-decrease}'). However, when the goal is strengthened to the critical top-level fault (reactor rupture), the planner reports no plan exists. RG or ERG can not find this plan because they consider fault as terminal events and safety measures are given as remediation suggestions.
The FT transitions required to escalate from `\texttt{fault-reactor-pressure-decrease}' to `\texttt{fault-reactor-rupture}' needs preconditions that are never established in the AG+FT model. %The engineering workstation is architecturally isolated; no attack action in the AG can reach it from the attacker's initial foothold under normal operating conditions. The system therefore appears safe with respect to reactor rupture under this model.

% We first analyze the system using an RG model that includes only the AG and FT layers (no ET).
% The defender's question in this scenario is:
% \emph{Does there exist an attack path that can trigger a fault condition?}
% Accordingly, we set the planning goal to reach a fault proposition of interest in the FT.
% The planner returns a feasible plan as shown in Figure~\ref{fig:fault-plan} that corresponds to an AG$\rightarrow$FT cascade:
% a sequence of cyber actions establishes the enabling conditions of a fault transition, which triggers the fault.

% \begin{figure}[!h]
%     \centering
%     \includegraphics[width=\linewidth]{figure/plan-FT.pdf}
%     \caption{Caption}
%     \label{fig:fault-plan}
% \end{figure}

% \begin{figure}[htbp]
% \centering
% \begin{minipage}{0.95\linewidth}
% \tiny
% \begin{verbatim}
% (scan-network-and-discover-hmi 192-dot-168-dot-90-dot-114 192-dot-168-dot-90-dot-5)
% (discover-plc-via-arp-spoofing 192-dot-168-dot-90-dot-5 192-dot-168-dot-95-dot-2)
% (connect-to-plc-bypassing-firewall 192-dot-168-dot-95-dot-2)
% (exploit-buffer-overflow-vulnerability-to-reset-pressure-sensor 192-dot-168-dot-95-dot-2)
% (fault-reactor-pressure-decrease )
% ; cost = 5 (unit cost)
% \end{verbatim}
% \end{minipage}
% \caption{Initial Attack Sequence/Path.}
% \label{fig:fault-plan}
% \end{figure}

\begin{figure}[htbp]
\centering
% \begin{tcolorbox}[
%     colback=white,
%     colframe=green,
%     boxrule=0.5pt,
%     arc=1mm,
%     width=0.95\linewidth,
%     boxsep=2pt,
%     left=2pt,
%     right=1pt,
%     top=.5pt,
%     bottom=.5pt
% ]
% \footnotesize
\begin{lstlisting}[style=pddlplanstyle]
((*@\bfseries scan-network-and-discover-hmi@*) 192-dot-168-dot-90-dot-114 
    192-dot-168-dot-90-dot-5)
((*@\bfseries discover-plc-via-arp-spoofing@*) 192-dot-168-dot-90-dot-5 
    192-dot-168-dot-95-dot-2)
((*@\bfseries connect-to-plc-bypassing-firewall@*) 192-dot-168-dot-95-dot-2)
((*@\bfseries exploit-buffer-overflow-vulnerability-to-reset-pressure-sensor@*) 
    192-dot-168-dot-95-dot-2)
((*@\bfseries fault-reactor-pressure-decrease@*) )
; cost = 5 (unit cost)
\end{lstlisting}
% \end{tcolorbox}
\caption{A PDDL representation of initial attack plan/path from IT cyber attacks to OT fault.}
\label{fig:fault-plan}
\end{figure}

Next, we integrate safety measures as the ET layer. Now, the defender's question becomes: \textit{``does reactor rupture become reachable once safety-response behavior is modeled''?}
The planner finds a plan that did not exist under AG+FT alone. %The cascade follows the pattern: $\text{AG} \rightarrow \text{FT} \rightarrow \text{ET} \rightarrow \text{AG} \rightarrow \text{FT} \rightarrow \cdots$. 
% The first four steps are identical to previous case. %— the attacker establishes an IT foothold and triggers the pressure fault. 
% The key difference is in what happens afterward. The fault activates the ET for \texttt{pressure-sensor-failure}. The planner evaluates barriers B1 and B2 (alarm visible, operator recognizes), then reaches barrier B3: PLC logic inspection from the engineering workstation. Evaluating B3 enables a new postconditions, which include \texttt{workstation-is-live} — a new host-level proposition that was absent from the previous model. This is the ET $\to$ AG transition where the safety procedure itself creates new cyber exposure.
% With the workstation now reachable, the planner executes two attack steps that are previously impossible, \textit{workstation discovery} and \textit{remote access}. The attacker uploads a malicious PLC control script through port 8080 of OpenPLC Editor, corrupting the pressure setpoint logic. Barriers B4–B6 evaluate under the now-compromised control environment, producing outcome O1 (\textit{Normal recovery/pressure recovers safely}). For the OT operator, following the safety procedures eliminates the faults, but they unknowingly also create a new cyber attack chain. The fault tree then escalates through \texttt{fault-relief-pressure-sensor-failure} and \texttt{fault-technical-protection-failure} to the top-level goal/fault, \texttt{fault-reactor-rupture} as shown in Figure \ref{fig:fault-plan-2}.
The first four steps are identical to the AG+FT case. The key difference is what happens after the initial pressure-related fault is triggered. If a fault is observed, operator typically follows a safety protocol. For this case, the fault, `\texttt{pressure-sensor-failure}', act as initial state ($i$) to activate safety sequences which presented as the event tree as shown in Figure~\ref{fig:fault-plan-2}. The planner
evaluates B1 and B2 (barrier's of ET), corresponding to alarm visibility and operator recognition, and then reaches B3, where the safety procedure requires PLC logic inspection from the engineering workstation.

This workstation was not previously connected to system's network. It has became live due to safety checks. From the operator's perspective, the `\texttt{pressure-sensor-failure}' response can still appear to reach the normal recovery outcome \(O_1\). However, evaluating B3 adds a new state in \frm\ ,`\texttt{workstation-is-live}'. This proposition was absent in the AG+FT model. But due to safety it now becomes a host-level
precondition for later workstation attack actions. This is the \textit{safety-induced later movement} of the adversary.
With the workstation now reachable, the planner finds a plan where the attacker compromises the workstation and uploads malicious PLC logic through the OpenPLC Editor interface. This corrupts the pressure setpoint logic and establishes a new fault-enabling condition. The later FT chain then escalates through `\texttt{fault-relief-pressure-sensor-failure}' and `\texttt{fault-technical-protection-failure}' to the top-level fault `\texttt{fault-reactor-rupture}' (Figure~\ref{fig:fault-plan-2}).

\begin{figure}[htbp]
\centering
\begin{lstlisting}[style=pddlplanstyle]
((*@\bfseries scan-network-and-discover-hmi@*) 192-dot-168-dot-90-dot-114 
    192-dot-168-dot-90-dot-5)
((*@\bfseries discover-plc-via-arp-spoofing@*) 192-dot-168-dot-90-dot-5 
    192-dot-168-dot-95-dot-2)
((*@\bfseries connect-to-plc-bypassing-firewall@*) 192-dot-168-dot-95-dot-2)
((*@\bfseries exploit-buffer-overflow-vulnerability-to-reset-pressure-sensor@*) 
    192-dot-168-dot-95-dot-2)
((*@\bfseries fault-rex-press-dcr @*))
((*@\bfseries initiate-event-tree-pressure-sensor-failure @*))
((*@\bfseries evaluate-barrier-b1-hmi-alarm-visibility @*))
((*@\bfseries evaluate-barrier-b2-operator-recognition @*))
((*@\tikzmark{sgboxA}\bfseries evaluate-barrier-b3-plc-logic-check-in-workstation @*))
((*@\bfseries discover-engineering-workstation@*) 192-dot-168-dot-95-dot-2 
    192-dot-168-dot-95-dot-5)(*@\tikzmark{sgboxB}@*)
((*@\bfseries gain-access-to-workstation@*) 192-dot-168-dot-95-dot-5)
((*@\bfseries upload-malicious-plc-pressure-setpoint@*) 192-dot-168-dot-95-dot-5)
((*@\bfseries evaluate-barrier-b4-corrective-action @*))
((*@\bfseries evaluate-barrier-b5-auto-protection @*))
((*@\bfseries evaluate-barrier-b6-isolation-prevention @*))
((*@\bfseries determine-safety-outcome @*))
((*@\bfseries fault-rlf-press-fail @*))
((*@\bfseries fault-tech-p-fail @*))
((*@\bfseries fault-rex-rup @*))
; cost = 19 (unit cost)
\end{lstlisting}
\begin{tikzpicture}[remember picture,overlay]
  \draw[red,densely dotted,thick,rounded corners=1mm]
    ($(sgboxA)+(-0.4em,1.1ex)$)
    rectangle
    ($(sgboxA|-sgboxB)+(0.72\linewidth,-0.3ex)$);
\end{tikzpicture}
\caption{Planner-generated SG cascade path for the TEP case. The dotted box
marks the transition where barrier b3 fails and the attacker pivots from the
PLC to the workstation.}
\label{fig:fault-plan-2}
\end{figure}

This analysis shows that \frm\ captures a cascade path that was absent in the prior models (RG and ERG). The missing link is not an additional exploit in the original attack graph. It is the \textit{safety-response-induced} exposure created during barrier B3. Analyzing the planner-generated trace now identifies a concrete hardening recommendation.
\textit{\textbf{The engineering workstation must not be brought online during safety checks without additional network controls or segmentation enforcement}}. \frm\ converts a subtle cross-layer interaction into a specific, actionable finding before a real adversary can exploit it. The mitigation process and operator's choices are discussed in the evaluation (\S \ref{sec:eval}).
%In the evaluation (\S \ref{sec:eval}) we show in detail how mitigation is captured as state proposition and applied, and discuss how operator's choice to implement a barrier effects the analysis and finding a plan or cascade path. %\frm\ has been additionally instantiated on an IEEE 14-bus power grid simulation using Node-Red (ref. Appendix~\ref{app:power-grid}). %, an automotive assembly line simulation using FactoryIO (full analysis of this domain is omitted due to space constraints but confirms that the framework instantiates for other critical infrastructures).

%SG does not merely append an event tree to an existing attack-fault path. Instead, the event tree changes the reachable cyber state space. In the AG+FT model, the precondition `\texttt{engineering-workstation-compromised}' is unreachable because the workstation remains isolated. In SG, barrier B3 introduces `\texttt{workstation-is-live}' proposition, enabling a new AG segment and therefore a later FT escalation. The actionable hardening recommendation is specific: the engineering workstation should not be made reachable during PLC logic inspection unless additional network controls or temporary isolation policies are enforced.

% ------------------------
% \section{Methodology}
% \input{40-Methodology}
% ------------------------
\section{Evaluation}
\label{sec:eval}

We validate the planner-generated \frm\ trace by replaying it step-by-step in the TEP test-bench, verifying that each abstract action — attack step, fault transition, or ET barrier evaluation — is executable in the deployed environment. The results are described in details in Appendix \ref{app:trace}. Validation evidence is drawn from network observations, protocol traces, reachability tests, workstation interaction logs, and process-state changes, compared against the corresponding state transitions in the compiled planning model.  We ran our experiment on Apple M3, 16 core macOS Tahoe, with 24 GB LPDDR5 memory.

\subsection{Operator's Decision \& Mitigation}
\noindent \underline{\textbf{Operator's Decision.}}
In \frm\ model, the outcome of each barrier can be either true or false. An operator has a choice to implement the barrier or not. Our model can work for any choice the operator makes and will update $\mu(b_i)$ of the current state. In \texttt{problem.pddl} the operator's decision for each barrier is represented as \texttt{(barier-b1) (barier-b2) (barier-b3) (barier-b4) (barier-b5) (barier-b6)}. If \texttt{(barier-b1)} is set to true as initial condition by operator (meaning operator decides to implement barrier B1), it will make proposition \texttt{(b1-alarm-visible)} true; if not \texttt{(b1-alarm-visible)} is false. For each decision based on operator, some propositions are true or false. For example if operator decides to perform all safety checks then outcome would be \texttt{(o1-normal-recovery)}. If the operator only evaluated  \texttt{(barier-b1), (barier-b5)} (others are false) it will generate an outcome \texttt{(o4-contamination)}. In the TEP running example, correctly executing the B3 PLC-logic check by operator is what brings the workstation online and opens the new attack surface. If the operator does not check the PLC logic (\texttt{(not (barier-b3))}) then $\mu(b3)$ = false, in which case there would not be any new attack as \texttt{workstation-is-live} proposition is false. %We will add a discussion in Section VI and APPENDIX A, on how operator’s choice to implement barriers  and attacker’s influence varies SG’s outcome. 

\begin{figure}[htbp]
\centering
\begin{lstlisting}[style=pddlplanstyle]
((*@\bfseries scan-network-and-discover-hmi@*) ip-192-dot-168-dot-90-dot-114 
    ip-192-dot-168-dot-90-dot-5)
((*@\bfseries discover-plc-via-arp-spoofing@*) ip-192-dot-168-dot-90-dot-5 
    ip-192-dot-168-dot-95-dot-2)
((*@\bfseries connect-to-plc-bypassing-firewall@*) 
    ip-192-dot-168-dot-95-dot-2)
((*@\bfseries exploit-buffer-overflow-vulnerability-to-reset-pressure-sensor@*) 
    ip-192-dot-168-dot-95-dot-2)
((*@\bfseries fault-rex-press-dcr @*))
((*@\bfseries initiate-event-tree-pressure-sensor-failure @*))
((*@\bfseries evaluate-barrier-b1-hmi-alarm-visibility @*))
((*@\bfseries evaluate-barrier-b2-operator-recognition @*))
((*@\bfseries evaluate-barrier-b3-plc-logic-check-in-workstation @*))
((*@\bfseries evaluate-barrier-b4-corrective-action @*))
((*@\bfseries evaluate-barrier-b5-auto-protection @*))
((*@\bfseries evaluate-barrier-b6-isolation-prevention @*))
((*@\bfseries determine-safety-outcome @*))
; cost = 13 (unit cost)
\end{lstlisting}
\caption{Planner-generated SG cascade path for the TEP case.}
\label{fig:operator-choice}
\end{figure}

As illustrated in Figure \ref{fig:operator-choice}, planner found a plan when operator does not perform PLC logic check (\texttt{(not (barier-b3))}), resulting in a normal recovery --- no new attack surface. This validates that operator decisions can influence the outcome, and \frm\ can model them. For single \texttt{initiate-event-tree-pressure-sensor-failure} with 6 barriers, we have run an iteration over $2^6$ operator's decisions and all possible goal state combinations. For the running case we have $2^5 - 2^{6}$ combinations of operator's choice for each fault and goal state that we run iteratively. 

\noindent \underline{\textbf{Mitigation.}}
After the plan analysis, it is found that bringing workstation live to the network for PLC check exposes workstation's IP to the attacker and a set of attacks could be launched (See steps 11-12 in Appendix \ref{app:trace}). The mitigation solution is isolating workstation for diagnosis or introducing security controls. By doing this it updates the planning compilation. In the previous \texttt{problem.pddl} initial state, the predicates \texttt{(node-exists ip-192-dot-168-dot-95-dot-5)} and \texttt{(connected-to ip-192-dot-168-dot-90-dot-114 ip-192-dot-168-dot-95-dot-5)} were true. Bringing up 192.168.95.5 into the network creates a path for attacker to upload malicious code to trigger \texttt{plc-pressure-setpoint} fault. With the security controls or isolation the workstation now becomes unreachable which results in both predicate to be false. With this update compilation and previous plan as input we ran VAL (validator) to find out the reasoning behind the failure.

\begin{figure}[htbp]
\centering
\begin{lstlisting}[style=valoutputstyle]
Plan failed because of unsatisfied precondition in:
((*@\bfseries discover-engineering-workstation@*) ip-192-dot-168-dot-95-dot-2 
    ip-192-dot-168-dot-95-dot-5)
Plan failed to execute

Plan Repair Advice:
((*@\bfseries discover-engineering-workstation@*) ip-192-dot-168-dot-95-dot-2 
    ip-192-dot-168-dot-95-dot-5) has an unsatisfied
    precondition at time 10
(*@\colorbox{yellow!45}{\ttfamily\scriptsize (Set (\textbf{node-exists} ip-192-dot-168-dot-95-dot-5) to true)}@*)
\end{lstlisting}
\caption{VAL's output for the failed \frm\ plan on the TEP case.}
\label{fig:plan-failure}
\end{figure}

As illustrated in Figure \ref{fig:plan-failure}, the action \texttt{discover- engineering-workstation} is failing because the workstation is not accessible by adversary (\texttt{node-exists ip-192-dot-168-dot-95-dot-5} is false) which stops subsequent cascade in the system.

\subsection{Scalability Analysis}

Although the SG framework is intended as an offline resiliency posture analysis platform, we still assessed the tractability for realistic IT-OT deployments. We evaluate on the  four core parameters of the SG -- attack nodes ($a_n$), fault nodes ($f_n$), event barriers per ET ($e_b$), and AG$\rightarrow$FT$\rightarrow$ET cascade cycles (\textit{seq}). We assume one event tree per fault. On the TEP instance, we compared Fast Downward with LM-cut against \textit{FF}, \textit{LAMA}, and
\textit{blind BFS}. LM-cut finds optimal plans fastest, since landmark extraction recovers the mandatory action sequence induced by the cascade structure. %Full results are available in the \textit{artifacts}\footnote{\textbf{Artifacts}. The test-bench will be open-sourced after the double blinded review for reproducibility. Everything else is available for review at https://anonymous.4open.science/r/Safety-Graph/.}.

The TEP instantiation is $a_n{=}19$, $f_n{=}32$, $e_b{=}6$ per ET ($\approx 190$ total), and $seq{=}2$, and it solves in 0.8 seconds. We sweep each parameter independently up to $a_n{=}200$, $f_n{=}200$, $e_b{=}100$ per ET, and $seq{=}30$. Two hundred attack nodes already exceed a large enterprise network, making the sweep aggressive relative to realistic deployment requirements.

\begin{figure}
    \centering
    \includegraphics[width=0.8\linewidth]{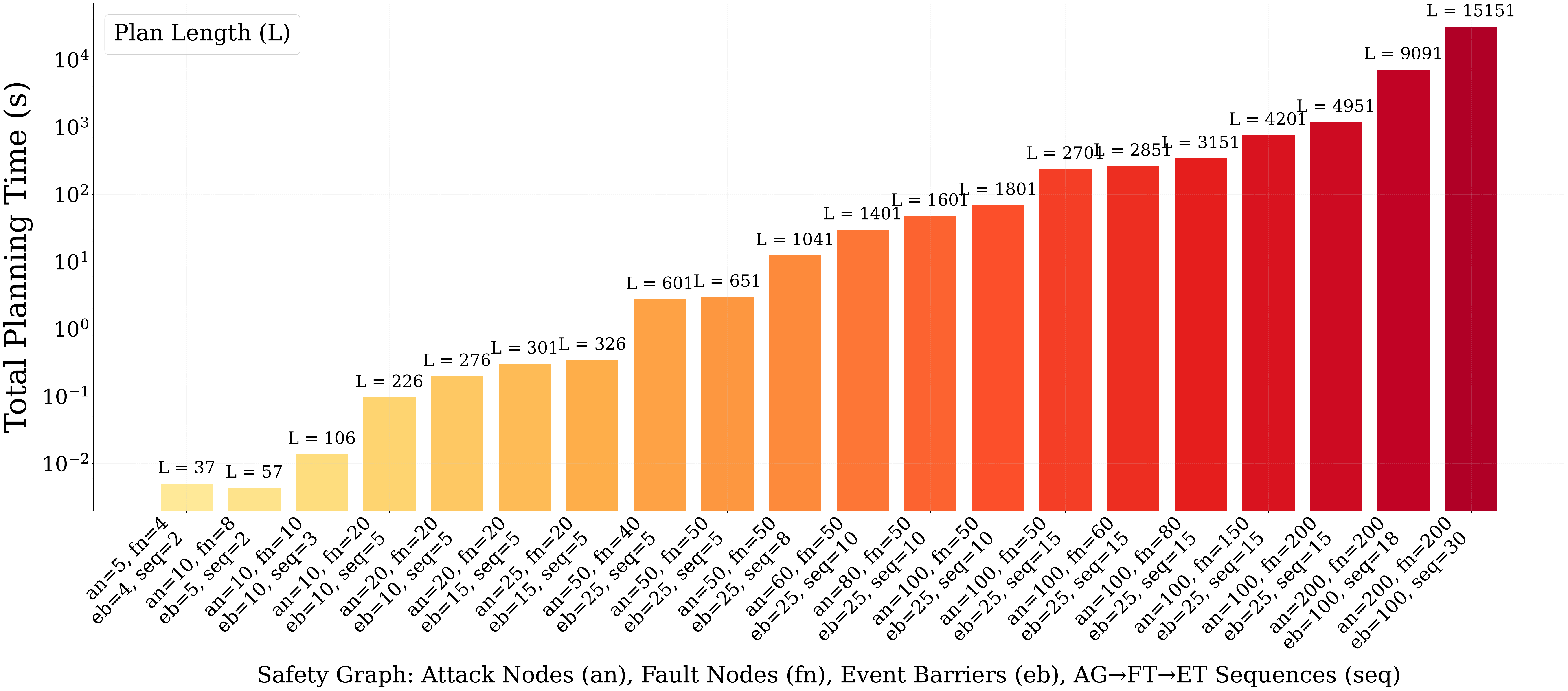}
    \caption{Planner time vs increasing the size of Safety Graph.}
    \label{fig:plan-time}
\end{figure}

The planner solves 192 configurations with planning time ranging from 0.004 to 1{,}592 seconds (Figure~\ref{fig:plan-time}). Peak memory stays at 415~MB across the full parameter range, so the encoding remains compact. Doubling
plan length increases planning time approximately eightfold. \textit{seq} dominates the cost because it multiplies the per-cycle contribution of every other parameter. Attack and fault nodes contribute additively within a cycle and scale near-identically, as the symmetry of plan length predicts. This behavior holds under the current SG assumption of deterministic barrier evaluation.

Configurations at five times the test-bench scale in every parameter solve within minutes. The practical tractability boundary is near $5{,}000$ plan steps under a 30-minute analysis budget. Practitioners should manage \textit{seq} first, as halving it reduces planning time by roughly $8\times$. Doubling $a_n$ or $f_n$ at fixed \textit{seq} costs $3$--$4\times$ and stays within interactive bounds. For very large deployments, decomposing the SG into sub-graphs with bounded \textit{seq} is the recommended mitigation. Reasoning over unbounded cascade cycles is left as future work.

% ------------------------
\section{Conclusion}
In this work, we have presented \frmtitle, a formal framework that unifies cyber attacks as Attack Graphs (AG), faults transitions as Fault Trees (FT), and safety measures as Event Trees (ET) over a shared propositional state space to capture the full bidirectional cascade between IT and OT layers in cyber-physical systems. 
%The central contribution of SG is modeling safety responses — prescribed, deterministic procedures that follow OT faults — as active participants in the attack surface rather than terminal events. 
%This OT→IT direction, where a physical fault triggers a safety response that unintentionally exposes new cyber preconditions enabling subsequent attacks, is absent from all prior frameworks.
We defined a sound and complete compilation of \frm\ to classical AI planning, reducing cross-layer cascade path discovery to a PSPACE-complete planning problem solvable automatically by a classical planner. The compilation preserves the full semantics of AG attack steps, FT fault transitions, and ET barrier evaluations as planning actions over a unified fluent space. On a simulated Tennessee Eastman Process test-bench, and
%instantiated from 32 domain expert-specified fault trees, approximately 190 barriers derived from certified IEC 61511 safety documentation, and a grounded attack graph from penetration testing evidence —
the planner is able to discover a critical cascade path traversing IT $\leftrightarrow$ OT. %Every step of the planner-generated trace was validated in the live test-bench along with 
SG produces hardening recommendations for cascade and verifies them by re-running the planner after applying recommended security control. We have also evaluated how operator's choices influence the outcome. Scalability analysis confirms tractability across configurations substantially larger than realistic ICS deployments. %, with planning time following a near-cubic relationship in plan length and memory remaining compact at 415 MB across all tested instances.
%In the future, we plan to extend \frm\ to nondeterministic barrier evaluation — where operator behavior may deviate from prescribed procedures — which would require conformant or probabilistic planning. 
Automated semantic alignment of propositions across AG, FT, and ET layers, currently requiring domain expert reconciliation, is a natural target for future task. Incremental re-analysis as the vulnerability surface evolves and multi-domain validation in other critical infrastructures are also directions for future work.
\label{sec:con}
\section*{Acknowledgment}
This work is partially supported by the U.S. Office of
Naval Research (ONR) under award N000142612041. The opinions, results, conclusions, and suggestions presented in this
work are those of the authors and do not necessarily represent those of the ONR or federal organizations and agencies.

% The authors would like to thank...

% trigger a \newpage just before the given reference
% number - used to balance the columns on the last page
% adjust value as needed - may need to be readjusted if
% the document is modified later
%\IEEEtriggeratref{8}
% The "triggered" command can be changed if desired:
%\IEEEtriggercmd{\enlargethispage{-5in}}

% references section

% can use a bibliography generated by BibTeX as a .bbl file
% BibTeX documentation can be easily obtained at:
% http://mirror.ctan.org/biblio/bibtex/contrib/doc/
% The IEEEtran BibTeX style support page is at:
% http://www.michaelshell.org/tex/ieeetran/bibtex/
%\bibliographystyle{IEEEtran}
% argument is your BibTeX string definitions and bibliography database(s)
%\bibliography{IEEEabrv,../bib/paper}
%
% <OR> manually copy in the resultant .bbl file
% set second argument of \begin to the number of references
% (used to reserve space for the reference number labels box)
% \begin{thebibliography}{1}

% \bibitem{IEEEhowto:kopka}
% H.~Kopka and P.~W. Daly, \emph{A Guide to \LaTeX}, 3rd~ed.\hskip 1em plus
%   0.5em minus 0.4em\relax Harlow, England: Addison-Wesley, 1999.

% \end{thebibliography}

\bibliographystyle{IEEEtran}
\bibliography{reference}

\appendices

\section{Planner Trace Execution}
\label{app:trace}
The planner-generated trace in Figure \ref{fig:fault-plan-2} contains attack, fault, and event steps. We evaluate these steps as follows:
%While some of these terminal steps are abstracted at the planner level, the test bench demonstrates the essential physical-security consequence required by the event-tree-enabled cyber action, which causes the system to progress beyond the initially reachable fault and into a more severe unsafe state that was not reachable in the AG+FT-only model.
\begin{figure}[!h]
    \centering
    \includegraphics[width=0.9\linewidth]{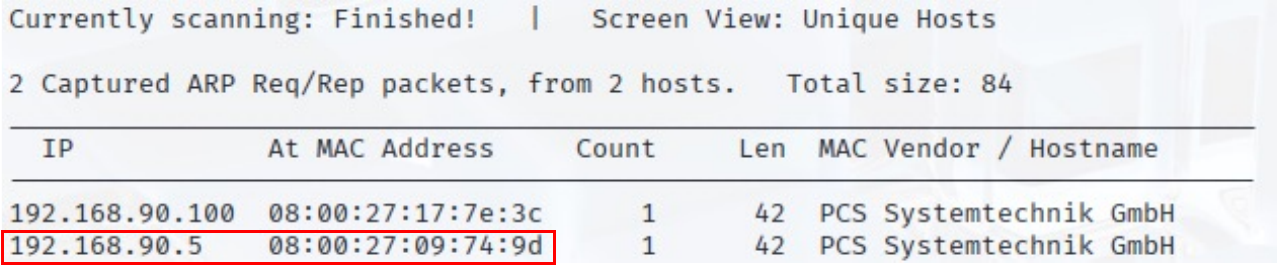}
    \caption{HMI discovery on the IT Subnet.}
    \label{fig:netdiscover}
\end{figure}

\noindent \underline{\textbf{Steps 1--2 (HMI and PLC discovery)}} Running \texttt{netdiscover} on the IT subnet (Figure \ref{fig:netdiscover}) confirms the HMI (\texttt{192.168.90.5}) is reachable from the attacker (\texttt{192.168.90.114}). ARP spoofing against the HMI reveals Modbus traffic to the PLC (\texttt{192.168.95.2}), exposing coil addresses and communication patterns (Figure~\ref{fig:wireshark}).

\begin{figure}
    \centering
    \includegraphics[width=\linewidth]{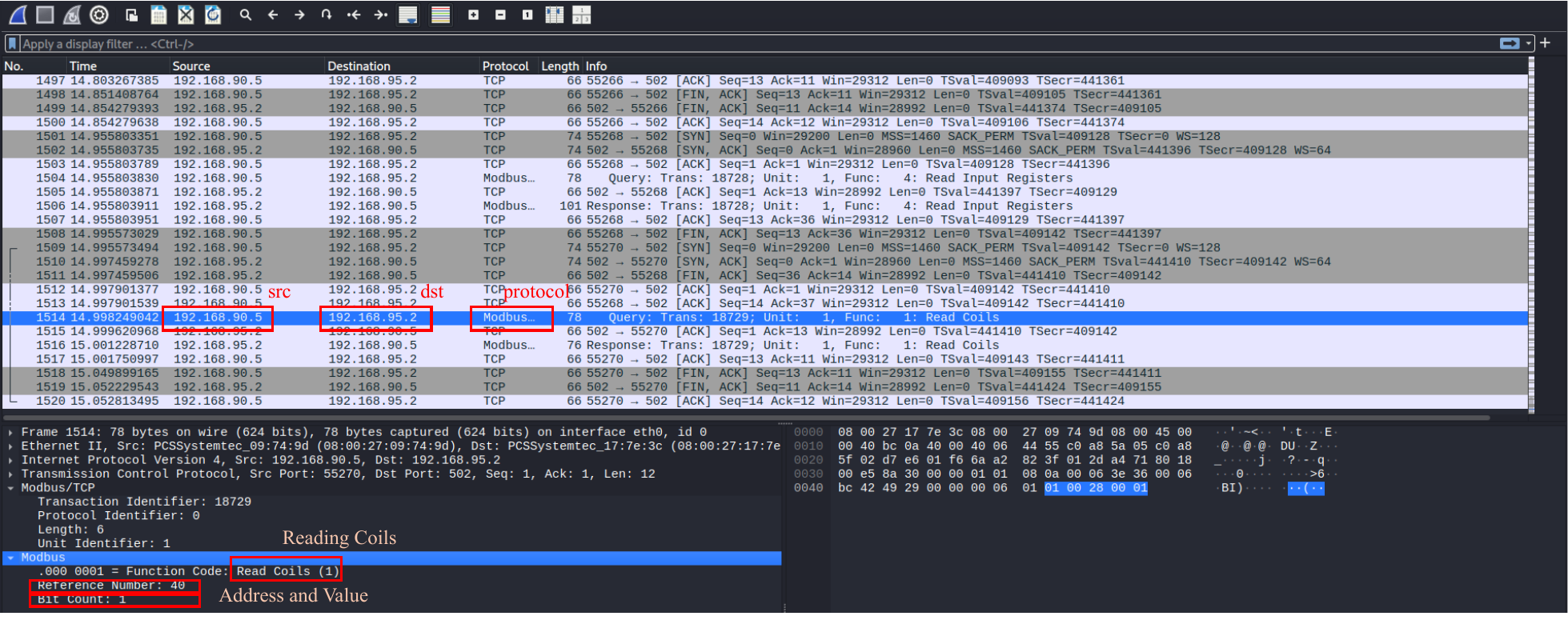}
    \caption{ARP spoofing and packet analysis using wireshark.}
    \label{fig:wireshark}
\end{figure}

\noindent \underline{\textbf{Step 3 (Firewall bypass)}} We use \texttt{Hydra} to crack the HMI password. SSH access reveals the pfSense gateway (\texttt{192.168.90.100}). Adding a static route to \texttt{192.168.95.0/24} via the firewall gives the attacker direct reachability to the PLC, validating the `\texttt{connect-to-plc-bypassing-firewall}' precondition.

\noindent \underline{\textbf{Step 4--5 (Modbus injection and fault initiation)}} An \texttt{nmap} scan identifies the OpenPLC service. A Metasploit Modbus write operation (Figure \ref{fig:metasploit}) resets the pressure coil, causing the main tank pressure to drop continuously (Figure \ref{fig:f1}), confirming the AG$\to$FT transition: `\texttt{fault- reactor-pressure-decrease}' is established in the physical process.

\begin{figure}
    \centering
    \includegraphics[width=0.9\linewidth]{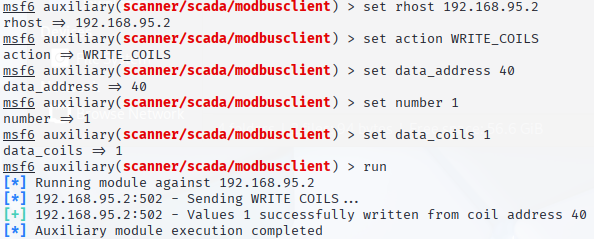}
    \caption{Modbus injection to PLC using Metasploit.}
    \label{fig:metasploit}
\end{figure}
% and observed that the main tank pressure value started to decreased rapidly as illustrated in Fig. \ref{fig:f1}.
\begin{figure}
    \centering
    \includegraphics[width=0.9\linewidth]{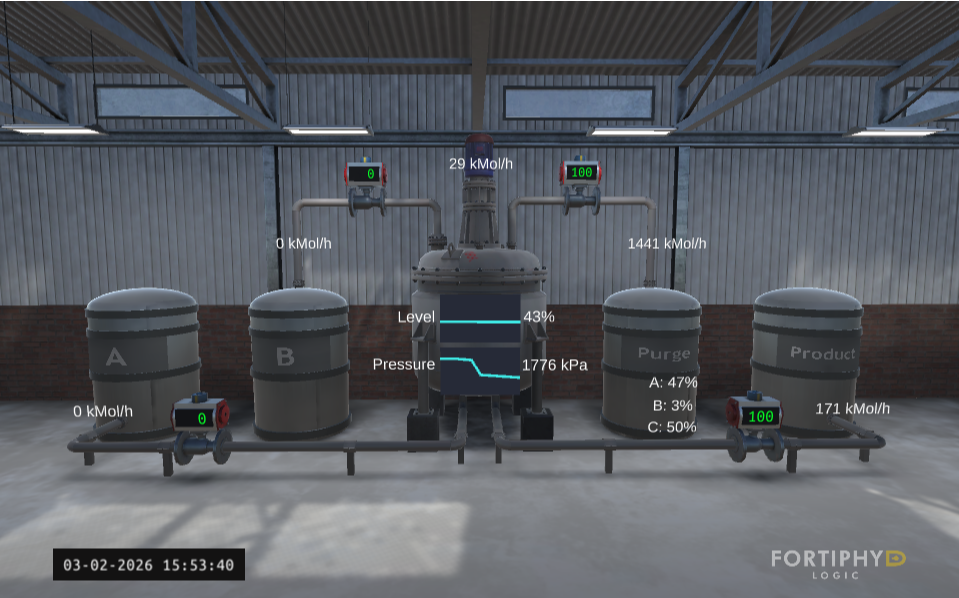}
    \caption{The pressure value of the main tank constantly decreasing.}
    \label{fig:f1}
\end{figure}

\noindent \underline{\textbf{Steps 6--9 (ET initiation and critical exposure)}} The abnormal pressure condition triggers the operator safety workflow. Barriers B1--B2 evaluate normally. At B3, `\texttt{evaluate-barrier-b3-plc-logic-check-in- workstation}', the safety procedure requires PLC logic inspection from the engineering workstation, bringing it online. This is the ET$\to$AG transition. The workstation, previously unreachable, enters the reachable state space exactly as modeled.

\noindent \underline{\textbf{Steps 10--12 (Workstation access and PLC corruption)}} The engineering workstation is located in the OT segment hosting OpenPLC Editor. Remote access is obtained and a malicious ladder logic file (\texttt{maliciousplc.st}) is uploaded via port 8080, corrupting the pressure setpoint. This validate our claim that safety response can create a cyber precondition that enables a previously unreachable attack.

\noindent \underline{\textbf{Steps 13--19 (Fault escalation to REX-RUP)}} %After the control logic is corrupted and even with barriers B4--B6 in place, it fails to yield a safe recovery. The process escalates through \texttt{fault-relief-pressure-sensor- failure} and \texttt{fault-technical-protection-failure} to \texttt{fault- reactor-rupture}, observed as pressure exceeding the threshold in the simulated plant (Figure~\ref{fig:f2}). %This final escalation was unreachable in the AG+FT-only model.
The pressure-sensor-failure event tree completes the operator-visible response, and the operator may interpret the local response as recovery. However, B3 has already exposed the engineering workstation, allowing the attacker to corrupt the PLC pressure setpoint logic. This corrupted logic establishes the new fault-enabling condition used by the subsequent FT transitions. The process then escalates through `\texttt{fault-relief-pressure-sensor-failure}' and `\texttt{fault-technical-protection-failure}' to `\texttt{fault-reactor-rupture}', observed as pressure exceeding the threshold in the simulated plant (Figure~\ref{fig:f2}).

\begin{figure}[!h]
    \centering
    \includegraphics[width=0.9\linewidth]{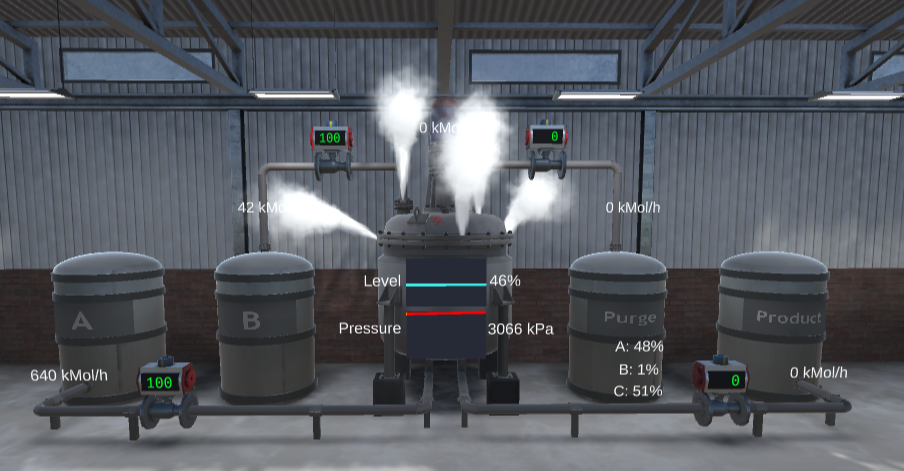}
    \caption{The pressure value of the main tank exceeds the threshold.}
    \label{fig:f2}
\end{figure}
Our evaluation confirms that the planner-generated \frm\ trace is not merely a modeling artifact. %Each major step in the returned plan can be grounded in the implemented environment through direct observation, network reconfiguration, control manipulation, workstation access, and process-state validation. 
Most importantly, the evaluation supports the central claim of \frm\ such that the critical failure path is not available in the previous models, but becomes realizable once safety-response behavior introduces new cyber exposure through the engineering workstation.

\section{Automated PDDL Generation}
\label{app:pddl}
% ============================================================
 
The compilation pipeline described in $\S$\ref{sec:SG} produces two PDDL files from the Safety Graph (SG) of the simulated TEP plant running example,  \texttt{domain.pddl}, which encodes the fluent space $\mathcal{F}^{SG}$ and action families $\mathcal{A}^{SG}$, and \texttt{problem.pddl}, which encodes the initial state $I^{SG}$, the object set, and the goal $G^{SG}$.  Both files are generated automatically by \texttt{pddl\_generator.py} (ref. Artifacts) from the parsed AG, FT, and ET (see Figure~\ref{fig:workflow}). %  This appendix walks through each section of the generated artifacts and identifies where every component of Definition~\ref{def:planning} is realised in concrete PDDL syntax.

\subsection{Attack Graph as Attack Actions}
\label{app:ag-actions}
% ------------------------------------------------------------

\begin{figure}[!h]
    \centering
    \includegraphics[width=0.9\linewidth]{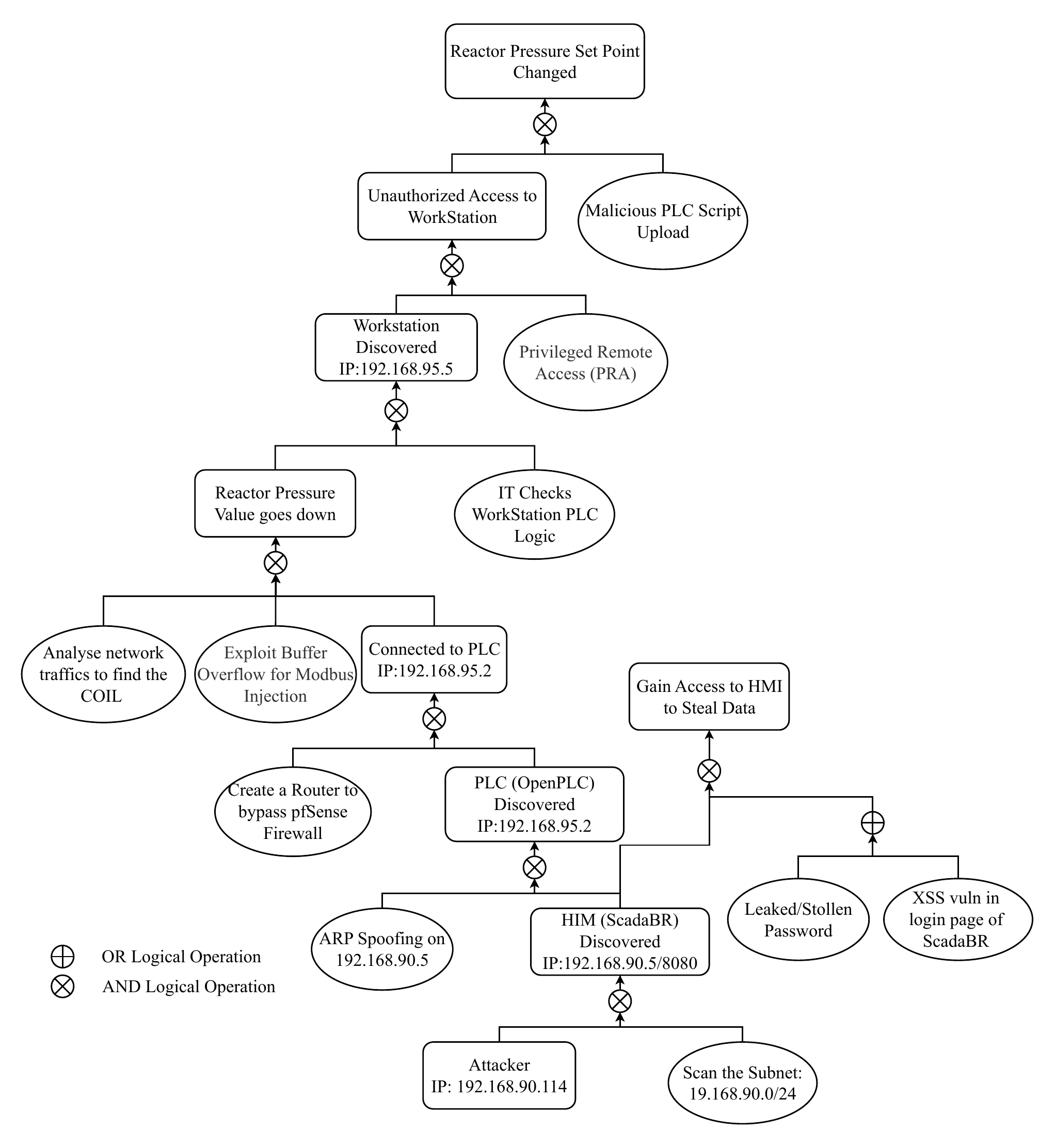}
    \caption{Attack Graph for simulated TEP plant.}
    \label{fig:ag}
\end{figure}

Each attack step in the AG becomes a PDDL action whose parameter signature is inferred from the step name and its proposition set. Actions whose source or destination involves two distinct hosts receive two node parameters \texttt{(?at~?to~-~node)} and actions confined to a single host receive one
parameter \texttt{(?at~-~node)}. Listing~\ref{lst:scan} shows an attack in AG encodes in PDDL.
 
% \paragraph{Two-hop action: \texttt{scan-network-and-discover -hmi}.}
% Listing~\ref{lst:scan} shows the first action of the cascade.
 
\begin{lstlisting}[caption={Attacker moves from initial
  foothold to HMI.},label={lst:scan}]
(:action scan-network-and-discover-hmi
    :parameters (?at ?to - node)
    :precondition (and
        (node-exists ?to)
        (at-node ?at)
        (connected-to ?at ?to)
        (attacker-scan-ip-and-port ?to)
    )
    :effect (and
        (attacker-discover-HMI ?to)
        (at-node ?to)
        (attacker-perform-arp-spoofing ?to)
    )
)
\end{lstlisting}
 
%The precondition set encodes $\mathit{pre}(a_e)$ from Definition~\ref{def:planning}: node existence, current attacker position (\texttt{at-node~?at}), a network edge (\texttt{connected-to~?at~?to}), and the scanning capability on the target.  %Note that \texttt{attacker-scan-ip-and-port} is set on the \emph{destination} node in \texttt{problem.pddl}, encoding the property that this host is scannable rather than that the attacker possesses a scanner; this is a deliberate modeling choice that allows different hosts to be marked as discoverable independently. 
%The effects encode $\mathit{add}(a_e)$: HMI discovery and ARP-spoofing capability are granted on the destination.  Movement is implicit: \texttt{(at-node~?to)} is added without deleting \texttt{(at-node~?at)}, so the attacker can be present at multiple nodes simultaneously once they have traversed an edge.
 % \paragraph{AG$\to$FT cross-layer link: \texttt{exploit-buffer-overflow}.}
Listing~\ref{lst:exploit} shows the attack action that enable the fault actions precondition.
 
\begin{lstlisting}[caption={Attack action establishes the fault action's precondition.},label={lst:exploit}]
(:action exploit-buffer-overflow-vulnerability-to-reset-pressure-sensor
    :parameters (?at - node)
    :precondition (and
        (attacker-connected-PLC ?at)
        (at-node ?at)
        (attacker-analyse-network-traffic-to-find-COIL-number ?at)
        (attacker-exploit-buffer-overflow-for-modbus-injection ?at)
    )
    :effect (and
        (Pressure-Sensor-Failure)
    )
)
\end{lstlisting}
 
% The effect \texttt{(Pressure-Sensor-Failure)} is a \emph{parameterless}
% FT proposition.  It appears immediately as a disjunct in the
% precondition of \texttt{fault-reactor-pressure-decrease} (Section~2 of
% the domain), realising the AG$\to$FT dependency without any explicit
% linking relation: the postcondition of an attack action serves directly
% as the enabling condition of a fault action, as specified in
% Definition~\ref{def:sg} and Theorem~\ref{thm:sound-complete}.
 
% ------------------------------------------------------------
\subsection{Fault Tree as Fault Actions}
\label{app:ft-actions}
% ------------------------------------------------------------

\begin{figure}
    \centering
    \includegraphics[width=0.8\linewidth]{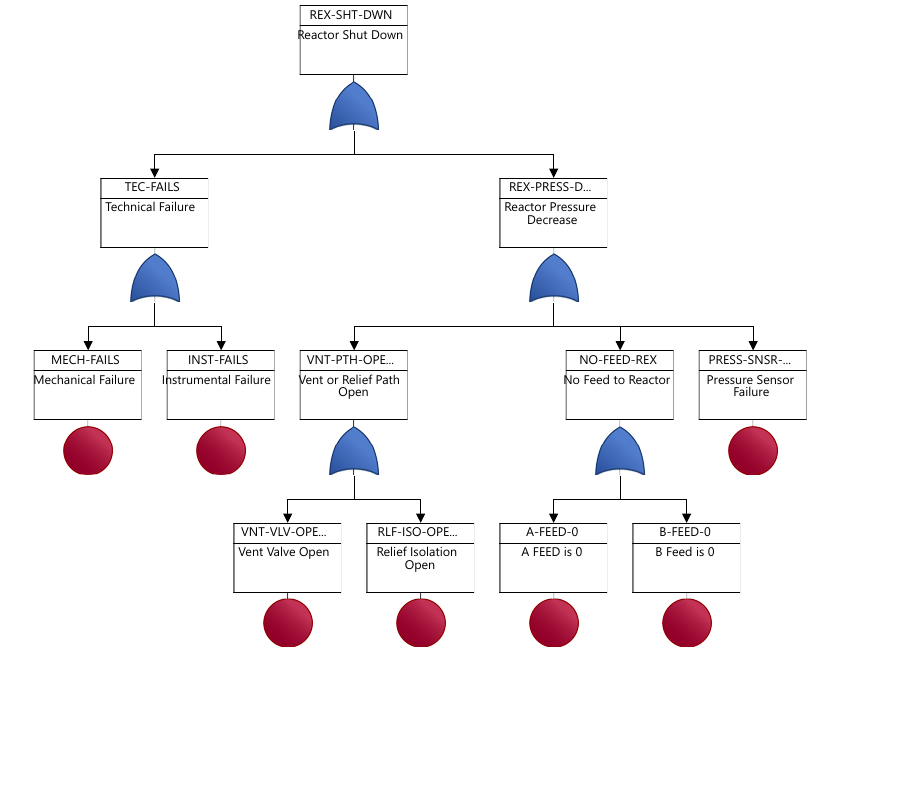}
    \caption{Fault Tree for reactor shut down (REX-SHT-DWN).}
    \label{fig:ft1}
\end{figure}

\begin{figure}
    \centering
    \includegraphics[width=\linewidth]{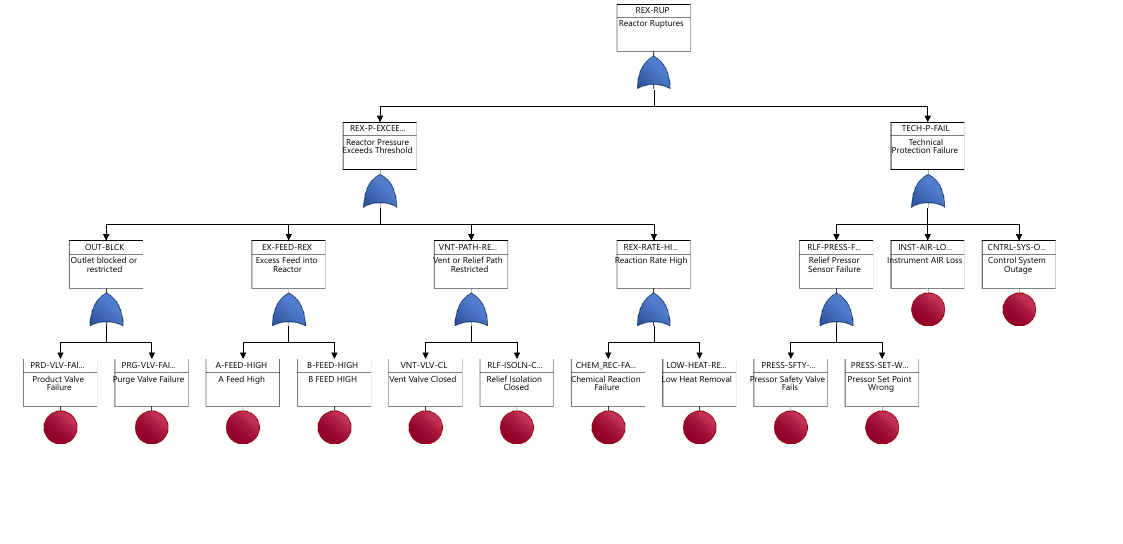}
    \caption{Fault Tree for reactor rupture (REX-RUP).}
    \label{fig:ft2}
\end{figure}

Each fault actions encodes one fault edge of the FT. The gate type (OR or AND) of the corresponding FT node determines the precondition shape.
 % \paragraph{OR-gate fault: \texttt{fault-reactor-pressure-decrease}.}
Listing~\ref{lst:ftpress} shows the a fault action triggered by the cyber attack or attack action.
 
\begin{lstlisting}[caption={FT action: OR gate over three enabling
  conditions, one of which is produced by the AG.},label={lst:ftpress}]
(:action fault-reactor-pressure-decrease
    :parameters ()
    :precondition (or
        (Vent-or-Relief-Path-Open)
        (No-Feed-to-Reactor)
        (Pressure-Sensor-Failure)   ; produced by AG exploit action
    )
    :effect (and
        (Reactor-Pressure-Decrease))
)
\end{lstlisting}
 
The OR precondition directly transcribes the OR gate in Figures~\ref{fig:ft1} and \ref{fig:ft2}.  The disjunct `\texttt{Pressure-Sensor-Failure}' is the AG$\to$FT link in the \texttt{domain.pddl}, where a cyber post-condition from Listing~\ref{lst:exploit} appears here as a physical enabling condition.
% \paragraph{Top-level fault with ET gate: \texttt{fault-reactor-rupture}.}
Listing~\ref{lst:ftrup} shows the fault action's pre- and post-conditions.

 % (:action fault-reactor-rupture
%     :parameters ()
%     :precondition (or
%         (Reactor-Pressure-Exceeds-Threshold)
%         (Technical-Protection-Failure)
%     )
%     :effect (and
%         (Reactor-Ruptures))
% )

\begin{lstlisting}[caption={FT action: top-level fault gated
  on ET completion.},label={lst:ftrup}]
(:action fault-relief-pressure-sensor-failure
    :parameters ()
    :precondition (and
        (event-tree-resolved)
        (or (Pressor-Safety-Valve-Fails) (Pressure-Set-Point-Wrong))
    )
    :effect (and
        (Relief-Pressor-Sensor-Failure))
)
\end{lstlisting}
 
% The precondition of \texttt{fault-relief-pressure-sensor-failure}
% requires \texttt{(event-tree-resolved)} in addition to a physical
% condition.  This is significant: \texttt{Reactor-Ruptures} is only
% reachable via the chain \texttt{fault-reactor-rupture} $\leftarrow$
% \texttt{Technical-Protection-Failure} $\leftarrow$
% \texttt{Relief-Pressor-Sensor-Failure}, whose enabling condition
% includes \texttt{event-tree-resolved}.  Because the planner cannot
% produce this proposition without executing \texttt{determine-safety-outcome}
% (see Section~\ref{app:et-actions}), the second fault cascade is
% structurally gated on ET completion.  This enforces the two-phase
% contrast of Section~\ref{sec:analysis}: without the ET layer,
% \texttt{event-tree-resolved} is never established and
% \texttt{Reactor-Ruptures} is unreachable.
 
% ------------------------------------------------------------
\subsection{Event Tree Actions}
\label{app:et-actions}
% ------------------------------------------------------------
 \begin{figure}[!h]
    \centering
    \includegraphics[width=\linewidth]{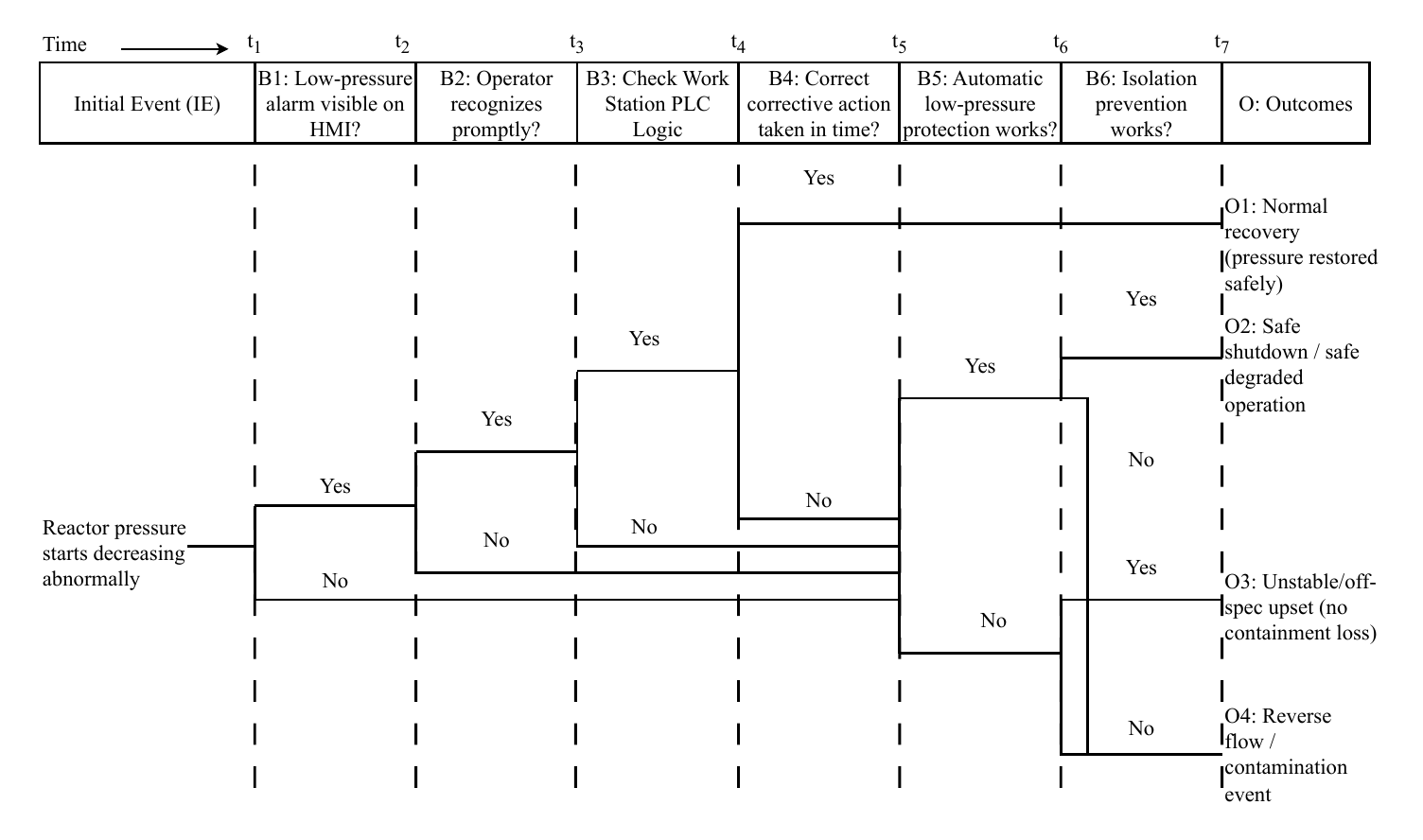}
    \caption{An Event Tree (\(E_{ie}\)) generated for a fault/Initial Event (\(i_e\)) `\texttt{pressure-sensor-failure}'.}
    \label{fig:et}
\end{figure}
% \paragraph{ET initiation.}
% The action \texttt{initiate-event-tree-pressure-sensor-failure}
% requires both \texttt{(Pressure-Sensor-Failure)} (AG postcondition)
% and \texttt{(Reactor-Pressor-Decrease)} (FT postcondition), encoding
% the FT$\to$ET activation dependency: the event tree fires only after
% the physical fault has propagated to the reactor pressure level.
 
% \paragraph{Sequential barrier evaluation.}
Each barrier $b_k$ compiles to a single action with two conditional effect branches.  Listing~\ref{lst:b1} shows B1 as a minimal example of the pattern.
 
\begin{lstlisting}[caption={ET action: B1 conditional effect structure
  encoding the yes/no barrier outcome (Definition~5).},label={lst:b1}]
(:action evaluate-barrier-b1-hmi-alarm-visibility
    :parameters ()
    :precondition (and (current-time-t1) (activate-event-tree))
    :effect (and
        (when (barrier-b1)
            (and (b1-alarm-visible) (b1-resolved)))
        (when (not (barrier-b1))
            (and (b1-resolved)))
        (not (current-time-t1))
        (current-time-t2)
    )
)
\end{lstlisting}
 
% The precondition \texttt{(current-time-t1)} enforces sequential
% evaluation: this action is only executable at time step~$t_1$,
% encoding the requirement $\mathit{at}(t_{k-1}) \in \mathit{pre}(b_k)$
% from Definition~\ref{def:et}.  Both branches unconditionally set
% \texttt{b1-resolved} and advance the clock to \texttt{current-time-t2},
% ensuring each barrier is evaluated exactly once regardless of outcome,
% consistent with the \emph{resolved} marker in Definition~\ref{def:et}.
% The decision proposition \texttt{(barier-b1)} is fixed true in
% \texttt{problem.pddl} (all barriers intact), reflecting the IEC~61511
% assumption that safety responses are prescribed in advance.
 
% \paragraph{ET$\to$AG cross-layer link: barrier B3.}
Listing~\ref{lst:b3} shows the pivotal barrier action.
 
\begin{lstlisting}[caption={ET barrier B3: `\textit{YES}' branch establishes the workstation-check proposition that enables next attack actions.},label={lst:b3}]
(:action evaluate-barrier-b3-plc-logic-check-in-workstation
    :parameters ()
    :precondition (and (current-time-t3) (b2-resolved))
    :effect (and
        (when (barrier-b3) ; YES: check performed
            (and (b3-plc-checked) (b3-resolved)))
        (when (not (barrier-b3)) ; NO: check skipped
            (and (b3-resolved))) 
        (not (current-time-t3))
        (current-time-t4)
    )
)
\end{lstlisting}

Listing~\ref{lst:outcome} shows the final ET action with the outcomes.
 
\begin{lstlisting}[caption={ET action: outcome resolution encodes the
  $\mathit{out}(\omega)$ function as four
  conditional effects.  The unconditional
  \texttt{(event-tree-resolved)} gates the second fault cascade.},
  label={lst:outcome}]
(:action determine-safety-outcome
    :parameters ()
    :precondition (and (current-time-t7) (b6-resolved)
                       (not (o1-normal-recovery))
                       (not (o2-safe-shutdown))
                       (not (o3-upset-no-loss))
                       (not (o4-contamination)))
    :effect (and
        ; O1: Normal recovery
        (when (and (b1-alarm-visible) (b2-operator-recognized)
                   (b3-plc-checked) (b4-corrective-in-time))
              (o1-normal-recovery))
        ; O2: Safe shutdown
        (when (and (not (b4-corrective-in-time))
                   (b5-auto-protection) (b6-isolation-prevention))
              (o2-safe-shutdown))
        ; O4: Contamination
        (when (and (not (b4-corrective-in-time))
                   (not (b6-isolation-prevention)))
              (o4-contamination))
        ; O3: Upset, no loss (catch-all)
        (when (and (not (o1-normal-recovery)) (not (o2-safe-shutdown))
                   (not (o4-contamination)))
              (o3-upset-no-loss))
        (event-tree-resolved)       ; gates second fault cascade
    )
)
\end{lstlisting}
 
Each \texttt{when} clause corresponds to one root-to-leaf path in the event tree of Figure~\ref{fig:et}.  The four conditional effects
jointly realize the outcome function $\mathit{out} : \{T,F\}^6 \to \mathcal{O}$ from Definition~\ref{def:et}.  The unconditional effect \texttt{(event-tree-resolved)} is produced regardless of which outcome branch fires, unlocking the second fault cascade. % in Section 5 of the domain.

\end{document}